\documentclass[11pt]{article}

\usepackage{amssymb}
\usepackage{amsfonts}
\usepackage{amsmath}
\usepackage{amsthm}
\usepackage[english]{babel}
\usepackage{latexsym}
\usepackage{color}
\usepackage{enumerate}
\usepackage{nicefrac}
\usepackage{mathrsfs}
\usepackage{comment}
\usepackage{bm}
\usepackage{bbm}
\usepackage[hmargin=2cm,vmargin=2cm]{geometry}

\usepackage{centernot}

\usepackage[affil-it]{authblk}
\theoremstyle{plain}
\newtheorem{theorem}{Theorem}[section]
\newtheorem{lemma}[theorem]{Lemma}
\newtheorem{proposition}[theorem]{Proposition}

\newtheorem{example}[theorem]{Example}
\newtheorem{remark}[theorem]{Remark}
\theoremstyle{definition}
\newtheorem{definition}[theorem]{Definition}

\theoremstyle{remark}

\numberwithin{equation}{section}

\newcommand{\one}{\mathbbm 1}

\newcommand{\agg}{S}

\newcommand{\ind}{\delta}

\newcommand{\ba}{\begin{array}{ll}}
\newcommand{\bal}{\begin{array}{ll}}
\newcommand{\ea}{\end{array}}

\newcommand{\E}{\mathbb{E}}

\newcommand{\probp}{\mathbb{P}}
\newcommand{\probq}{\mathbb{Q}}
\newcommand{\R}{\mathbb{R}}
\newcommand{\N}{\mathbb{N}}

\newcommand{\cF}{{\mathcal{F}}}

\newcommand{\cB}{\mathop {\rm bar}}
\newcommand{\cS}{{\mathcal{S}}}
\newcommand{\cA}{\mathcal{A}}
\newcommand{\cC}{\mathcal{C}}
\newcommand{\cD}{\mathcal{D}}

\newcommand{\cK}{\mathcal{K}}

\newcommand{\cL}{\mathcal{L}}
\newcommand{\cM}{\mathcal{M}}

\newcommand{\cQ}{\mathcal{Q}}
\newcommand{\cU}{\mathcal{U}}

\newcommand{\cX}{{\mathcal{X}}}

\newcommand{\cl}{\mathop{\rm cl}\nolimits}

\newcommand{\VaR}{\mathop {\rm VaR}\nolimits}

\newcommand{\ES}{\mathop {\rm ES}\nolimits}

\def\keywords{\vspace{.5em}
{\noindent\textbf{Keywords}:\,\relax%
}}

\makeatletter
\def\@fnsymbol#1{\ensuremath{\ifcase#1\or 1\or 2\or 3\or 4\or 5\or 6\or 7\or 8\else\@ctrerr\fi}}
\makeatother

\newcommand{\baf}{\begin{array}{llcl}}

\newcommand{\ucl}{\mathop{\rm usc}\nolimits}

\newcommand{\dom}{\mathop{\rm dom}\nolimits}
\newcommand{\cE}{{\mathcal{E}}}

\begin{document}

\title{
Dual representations for systemic risk measures based on acceptance sets
}

\author{
\sc{Maria Arduca}
}
\affil{Department of Statistics and Quantitative Methods, University of Milano-Bicocca}
\author{
\sc{Pablo Koch-Medina},
\sc{Cosimo Munari}
}
\affil{Center for Finance and Insurance and Swiss Finance Institute, University of Zurich, Switzerland}

\date{\today}

\maketitle

\begin{abstract}
\noindent
We establish dual representations for systemic risk measures based on acceptance sets in a general setting. We deal with systemic risk measures of both ``first allocate, then aggregate'' and ``first aggregate, then allocate'' type. In both cases, we provide a detailed analysis of the corresponding systemic acceptance sets and their support functions. The same approach delivers a simple and self-contained proof of the dual representation of utility-based risk measures for univariate positions.
\end{abstract}


\parindent 0em \noindent


\section{Introduction}

The study of risk measures for multivariate positions was first developed by Burgert and R\"{u}schendorf~\cite{BurgertRueschendorf2006}, R\"{u}schendorf~\cite{Rueschendorf2006}, Ekeland and Schachermayer~\cite{EkelandSchachermayer2011}, and Ekeland et al.~\cite{EkelandGalichonHenry2012} and extended to the set-valued case by Jouini et al.~\cite{JouiniMeddebTouzi2004}, Hamel and Heyde~\cite{HamelHeyde2010}, Hamel et al.~\cite{HamelHeydeRudloff2011}, and Molchanov and Cascos~\cite{MolchanovCascos2016}. In this literature, multivariate positions are typically interpreted as (random) portfolios of financial assets. In recent years, there has been significant interest in extending the theory of risk measures to a systemic risk setting, in which multivariate positions represent the (random) vector of future capital positions, i.e.\ assets net of liabilities, of financial institutions. In this setting, one considers a system of $d$ financial institutions whose respective capital positions at a fixed future date is represented by a random vector
\[
X=(X_1,\dots,X_d).
\]
The bulk of the literature assumes that a macroprudential regulator specifies an ``aggregation function''
\[
\Lambda:\R^d\to\R
\]
by means of which the system is summarized into a single (univariate) aggregated position $\Lambda(X)$. The simplest aggregation function is given by $\Lambda(x)=\sum_{i=1}^dx_i$ and corresponds to aggregating the entire system into a single consolidated balance sheet. The regulator also specifies a set $\cA$ of ``acceptable'' aggregated positions: The level of systemic risk of the financial system is deemed acceptable whenever $\Lambda(X)$ belongs to the prescribed acceptance set $\cA$. Two main classes of systemic risk measures based on aggregation functions and acceptance sets have been studied in the literature.

\smallskip

A first branch of the literature adopts a so-called ``first allocate, then aggregate'' approach, which is the macroprudential equivalent of the fundamental idea introduced in the context of  microprudential regulation by Artzner et al.~\cite{ArtznerDelbaenEberHeath1999}: To ensure that the financial system has an acceptable level of systemic risk, the macroprudential regulator can require each of the member institutions to raise a suitable amount of capital. Such a requirement is represented by a vector $m=(m_1,\dots,m_d)\in\R^d$, where $m_i$ corresponds to the amount of capital raised by institution $i$. This leads to a systemic risk measure of the form
\[
\rho(X) = \inf\left\{\sum_{i=1}^dm_i \,; \ m\in\R^d, \ \Lambda(X+m)\in\cA\right\}.
\]
The quantity $\rho(X)$ corresponds to the minimum amount of aggregate capital that needs to be injected into the financial system to ensure acceptability. This type of systemic risk measures has been studied in Feinstein et al.~\cite{FeinsteinRudloffWeber2017}, Armenti et al.~\cite{ArmentiCrepeyDrapeauPapapantoleon2018}, Ararat and Rudloff~\cite{AraratRudloff}, and Biagini et al.~\cite{BiaginiFouqueFrittelliMeyerBrandis2019a, BiaginiFouqueFrittelliMeyerBrandis2019b} (where random allocations of the aggregate capital requirement are also considered).

\smallskip

A second branch of the literature advocates a ``first aggregate, then allocate'' approach and studies systemic risk measures of the form
\[
\widetilde{\rho}(X) = \inf\{m\in\R \,; \ \Lambda(X)+m\in\cA\}.
\]
In this case, the quantity $\widetilde{\rho}(X)$ represents the minimal amount of capital that has to be added to the aggregated position to reach acceptability. In contrast to $\rho$, the operational interpretation of $\widetilde{\rho}$ is not straightforward since it is unclear how much each of the member institutions should contribute to the aggregate amount of required capital or, if the outcome of the above risk measure is interpreted as a bail-out cost, which institution should obtain which amount. Such systemic risk measures have been studied in Chen et al.~\cite{ChenIyengarMoallemi2013}, Kromer et al.~\cite{KromerOverbeckZilch2016}, and Ararat and Rudloff~\cite{AraratRudloff}.

\smallskip

The main objective of this note is to establish dual representations for the above systemic risk measures in a general setting with special emphasis on systemic risk measures of the ``first allocate, then aggregate'' type. By doing so, we provide a unifying perspective on the existing duality results in the literature. More precisely, we consider an arbitrary probability space $(\Omega,\cF,\probp)$ and assume that the multivariate positions belong to a space of $d$-dimensional random vectors $\cX$ that is in duality with another space of $d$-dimensional random vectors $\cX'$ through the pairing
\[
(X,Z) \mapsto \sum_{i=1}^d\E_\probp[X_iZ_i]
\]
for $X\in\cX$ and $Z\in\cX'$. This setup is general enough to cover all the interesting examples encountered in the literature. Dual representations for $\rho$ have been studied by Armenti et al.~\cite{ArmentiCrepeyDrapeauPapapantoleon2018} and Biagini et al.~\cite{BiaginiFouqueFrittelliMeyerBrandis2019b} in the setting of Orlicz hearts and acceptance sets based on (multivariate) utility functions and by Ararat and Rudloff~\cite{AraratRudloff} in the setting of bounded random variables with only mild restrictions on the acceptance set. The strategy in~\cite{ArmentiCrepeyDrapeauPapapantoleon2018} is to apply Lagrangian techniques while that of~~\cite{AraratRudloff} is to rely on the dual representation of $\widetilde{\rho}$, which is tackled by using Fenchel-Moreau techniques for composed maps. Similarly to~\cite{BiaginiFouqueFrittelliMeyerBrandis2019b}, the point of departure of this note is to observe that $\rho$ can be written as
\[
\rho(X) = \inf\{\pi(m) \,; \ m\in\R^d, \ X+m\in\Lambda^{-1}(\cA)\}
\]
where the ``acceptance set'' $\Lambda^{-1}(\cA)$ and the ``cost functional'' $\pi:\R^d\to\R$ are given by
\[
\Lambda^{-1}(\cA)=\{X\in\cX \,; \ \Lambda(X)\in\cA\}, \ \ \ \ \pi(m) = \sum_{i=1}^dm_i.
\]
This shows that $\rho$ belongs to the class of ``multi-asset risk measures'' introduced in Frittelli and Scandolo~\cite{FrittelliScandolo2006} and thoroughly studied in Farkas et al.~\cite{kochmunari_multiasset}. Under suitable conditions on $\Lambda$ and $\cA$, the general duality results obtained in those papers yield the representation
\[
\rho(X) = \sup\left\{\sigma(\probq_1,\dots,\probq_d)-
\sum_{i=1}^d\E_{\probq_i}[X_i] \,; \ \probq_1,\dots,\probq_d\ll\probp, \ \bigg(\frac{d\probq_1}{d\probp},\dots,\frac{d\probq_d}{d\probp}\bigg)\in\cX'\right\}
\]
where the objective function is given by
\[
\sigma(\probq_1,\dots,\probq_d)=\inf_{X\in\Lambda^{-1}(\cA)}\sum_{i=1}^d\E_{\probq_i}[X_i].
\]
The map $\sigma$ corresponds to the (lower) support function of the systemic acceptance set $\Lambda^{-1}(\cA)$ and plays a fundamental role in the dual representation. The main technical contribution of this note is to provide a detailed analysis of these objects. In particular, we devote some effort to obtain a more explicit description of $\sigma$ in terms of the primitives $\Lambda$ and $\cA$. It is worth noting that, in the systemic setting, the closedeness of the acceptance set $\Lambda^{-1}(\cA)$ does not necessarily imply the lower semicontinuity of $\rho$, which is a necessary condition for $\rho$ to admit a dual representation. Hence, it is important to provide conditions on the primitives $\Lambda$ and $\cA$ to ensure that $\rho$ is lower semicontinuous in the first place. To this effect, we rely on the abstract results in Farkas et al.~\cite{kochmunari_multiasset}, where an effort was made to derive properties of risk measures, such as lower semicontinuity, from the properties of the underlying acceptance sets. In particular, the duality results for acceptance sets obtained there build the starting point for our analysis of the systemic acceptance set $\Lambda^{-1}(\cA)$.

\smallskip

On the other side, dual representations for $\widetilde{\rho}$ have been studied by Chen et al.~\cite{ChenIyengarMoallemi2013} in a finite-dimensional setting, by Ararat and Rudloff~\cite{AraratRudloff} in the setting of bounded random vectors, and by Kromer et al.~\cite{KromerOverbeckZilch2016} at our level of generality. In line with those papers, our starting point is to observe that $\widetilde{\rho}$ can be expressed as a standard cash-additive risk measure (namely the cash-additive risk measure associated with $\cA$, which we denote by $\rho_\cA$) applied to the aggregated univariate position as
\[
\widetilde{\rho}(X) = \rho_\cA(\Lambda(X)).
\]
However, instead of working out the Fenchel-Moreau representation of a composition of maps, we exploit the standard dual representation for $\rho_\cA$ to derive the desired representation in a direct way. In this case, conditions to ensure the lower semicontinuity of $\widetilde{\rho}$ are also easier to formulate.

\smallskip

Finally, to illustrate the convenience of our approach to duality based on acceptance sets, we provide a simple and self-contained proof of the dual representation for utility-based risk measures for univariate positions, which can be viewed as special systemic risk measures where $d=1$ and $\Lambda$ is a von Neumann-Morgenstern utility function.


\section{The setting}
\label{sect: rho introduced}

In this section, we recall the necessary terminology and notation from probability theory and convex analysis and describe our setting. For a general presentation of convex duality we refer to Aliprantis and Border~\cite{AliprantisBorder2006} and Z\u{a}linescu~\cite{Zalinescu2002}.


\subsection{Basic notation}

Throughout, we fix a probability space $(\Omega,\cF,\probp)$. We denote by $L^0$ the vector space of (equivalence classes with respect to almost-sure equality of) random variables, i.e.\ Borel-measurable functions $X:\Omega\to\R$. As usual, we do not distinguish explicitly between an equivalence class and any of its representatives. For every $p\in[1,\infty]$ we denote by $L^p$ the space of $p$-integrable random variables if $p<\infty$ and the space of bounded random variables if $p=\infty$.

\smallskip

A (nonzero) linear subspace $\cL\subset L^0$ is said to be {\em admissible} if it is a Banach lattice (with respect to the almost-sure partial order) such that $L^\infty\subset\cL\subset L^1$. In this case, we set
\[
\cL' := \{Z\in L^0 \,; \ \E\left[|XZ|\right]<\infty, \ \forall X\in\cL\},
\]
where $\E$ denotes the expectation with respect to $\probp$. Note that we always have $L^\infty\subset\cL'\subset L^1$.

\begin{example}
The class of admissible spaces contains Orlicz spaces, which include the standard examples encountered in the literature. A nonconstant function $\Phi:[0,\infty)\to[0,\infty]$ is said to be an Orlicz function if it is convex, nondecreasing, right-continuous at $0$, and satisfies $\Phi(0)=0$. The convex conjugate of $\Phi$ is the function $\Phi^\ast:[0,\infty)\to[0,\infty]$ given by
\[
\Phi^\ast(t) := \sup_{s\in[0,\infty)}\{st-\Phi(s)\}.
\]
It is easy to see that $\Phi^\ast$ is also an Orlicz function. The Orlicz space associated with $\Phi$ is
\[
L^\Phi := \left\{X\in L^0 \,; \ \E\left[\Phi\left(\frac{|X|}{\lambda}\right)\right]<\infty \ \mbox{for some $\lambda\in(0,\infty)$}\right\}.
\]
The corresponding Orlicz heart is defined by
\[
H^\Phi := \left\{X\in L^0 \,; \ \E\left[\Phi\left(\frac{|X|}{\lambda}\right)\right]<\infty \ \mbox{for every $\lambda\in(0,\infty)$}\right\}.
\]
These spaces are Banach lattices with respect to the Luxemburg norm
\[
\|X\|_\Phi := \inf\left\{\lambda\in(0,\infty) \,; \ \E\left[\Phi\left(\frac{|X|}{\lambda}\right)\right]\leq1\right\}.
\]
The norm dual of $L^\Phi$ cannot be identified with a subspace of $L^0$ in general. However, the norm dual of $H^\Phi$ can always be identified with $L^{\Phi^\ast}$ provided that $\Phi$ is finite valued (otherwise $H^\Phi=\{0\}$). We say that $\Phi$ satisfies the $\Delta_2$ condition if there exist $r\in(0,\infty)$ and $k\in(0,\infty)$ such that $\Phi(2s)<k\Phi(s)$ for every $s\in[r,\infty)$. It is well-known that, if $\Phi$ is $\Delta_2$, then $L^\Phi=H^\Phi$. In a nonatomic setting the converse implication also holds.

\smallskip

Every $L^p$ space, with $p\in[1,\infty]$, can be viewed as an Orlicz space. Indeed for every $p\in[1,\infty)$ we have $L^\Phi=L^p$ if we set $\Phi(s)=s^p$ for every $s\in[0,\infty)$. In this case, the Luxemburg norm coincides with the standard $L^p$ norm and we have $L^\Phi=H^\Phi$. Moreover, we have $L^\Phi=L^\infty$ if we define $\Phi(s)=0$ for $s\in[0,1]$ and $\Phi(s)=\infty$ otherwise. In this case, the Luxemburg norm coincides with the standard $L^\infty$ norm and we have $H^\Phi=\{0\}$.

\smallskip

The following statements hold; see e.g.\ Edgar and Sucheston~\cite{EdgarSucheston1992} or Meyer-Nieberg~\cite{MeyerNieberg1991}:
\begin{enumerate}[(1)]
  \item $\cL=L^\Phi$ is admissible and $\cL'=L^{\Phi^\ast}$;
  \item $\cL=H^\Phi$ is admissible if $\Phi$ is finite valued, in which case $\cL'=L^{\Phi^\ast}$;
  \item $\cL=L^p$ is admissible if $p\in[1,\infty]$, in which case $\cL'=L^{\frac{p}{p-1}}$ (with the convention $\frac{1}{0}=\infty$ and $\frac{\infty}{\infty}=1$).
\end{enumerate}
\end{example}

\medskip

Fix $m\in\N$. We always consider the standard inner product $\langle\cdot,\cdot\rangle:\R^m\times\R^m\to\R$ defined by
\[
\langle a,b\rangle := \sum_{i=1}^ma_ib_i.
\]
We denote by $L^0(\R^m)$ the vector space of all $m$-dimensional random vectors. We say that a linear subspace $\cL\subset L^0(\R^m)$ is {\em admissible} whenever
\[
\cL = \cL_1\times\cdots\times\cL_m
\]
with admissible $\cL_1,\dots,\cL_m\subset L^0$. Note that, being the product of Banach lattices, the space $\cL$ is also a Banach lattice. In particular, the lattice operations on $\cL$ are understood component by component. As above, we define
\[
\cL' := \cL'_1\times\cdots\times\cL'_m.
\]
The pair $(\cL,\cL')$ is equipped with the bilinear form $(\cdot\vert\cdot):\cL\times\cL'\to\R$ given by
\[
(X\vert Z) := \E[\langle X,Z\rangle] = \sum_{i=1}^m\E[X_iZ_i].
\]
The coarsest topology on $\cL$ making the linear functional $X\mapsto(X\vert Z)$ continuous for every $Z\in\cL'$ is denoted by $\sigma(\cL,\cL')$. Similarly, the coarsest topology on $\cL'$ making the linear functional $Z\mapsto(X\vert Z)$ continuous for every $X\in\cL$ is denoted by $\sigma(\cL',\cL)$. Equipped with these topologies, $\cL$ and $\cL'$ are locally-convex topological vector spaces (which are also Hausdorff because the above form is separating).

\smallskip

The following standard notions from convex analysis will be freely used throughout the note. For a subset $\cS\subset L^0(\R^m)$ we denote by $\cS_+$, respectively $\cS_{++}$, the set of all random vectors in $\cS$ with components that are almost-surely nonnegative, respectively strictly-positive. The domain of finiteness of a map $f:\cL\to[-\infty,\infty]$ is defined by
\[
\dom(f) := \{X\in\cL \,; \ f(X)\in\R\}.
\]
The {\em (lower) support function} of a (nonempty) set $\cA\subset\cL$ is the map $\sigma_\cA:\cL'\to[-\infty,\infty)$ defined by
\[
\sigma_\cA(Z) := \inf_{X\in\cA}\E[\langle X,Z\rangle].
\]
Its domain of finiteness is denoted by $\cB(\cA)$ and called the {\em barrier cone}, i.e.
\[
\cB(\cA) := \dom(\sigma_\cA) = \{Z\in\cL' \,; \ \sigma_\cA(Z)>-\infty\}.
\]
If $\cA$ is $\sigma(\cL,\cL')$-closed and convex, the Hahn-Banach Theorem yields the representation
\begin{equation}
\label{eq: external characterization}
\cA = \bigcap_{Z\in\cL'}\{X\in\cL \,; \ \E[\langle X,Z\rangle]\geq\sigma_\cA(Z)\} = \bigcap_{Z\in\cB(\cA)}\{X\in\cL \,; \ \E[\langle X,Z\rangle]\geq\sigma_\cA(Z)\}.
\end{equation}
We say that $\cA$ is {\em monotone} if $\cA+\cL_+\subset\cA$. In this case, we have $\cB(\cA)\subset\cL'_+$.

\smallskip

Consider a nonconstant map $f:\cL\to(-\infty,\infty]$. The {\em convex conjugate} of $f$ is the map $f^\ast:\cL'\to(-\infty,\infty]$ given by
\[
f^\ast(Z) := \sup_{X\in\cL}\{\E[\langle X,Z\rangle]-f(X)\}.
\]
The Fenchel-Moreau Theorem states that, if $f$ is convex and $\sigma(\cL,\cL')$-lower semicontinuous, then
\[
f(X) = \sup_{Z\in\cL'}\{\E[\langle X,Z\rangle]-f^\ast(Z)\}
\]
for every $X\in\cL$. We will find it convenient to use the concave counterpart to convex duality. Here, consider a nonconstant map $g:\cL\to[-\infty,\infty)$. The {\em concave conjugate} of $g$ is the map $g^\bullet:\cL'\to[-\infty,\infty)$ defined by
\[
g^\bullet(Z) := \inf_{X\in\cL}\{\E[\langle X,Z\rangle]-g(X)\}=-(-g)^\ast(-Z).
\]
Hence, the Fenchel-Moreau Theorem implies that, if $g$ is concave and $\sigma(\cL,\cL')$-upper semicontinuous, then
\[
g(X) = \inf_{Z\in\cL'}\{\E[\langle X,Z\rangle]-g^\bullet(Z)\}
\]
for every $X\in\cL$. The {\em indicator function} of a set $\cA\subset\cL$ is the map $\ind_\cA:\cL\to[0,\infty]$ given by
\[
\ind_\cA(X):=
\begin{cases}
0 & \mbox{if} \ X\in\cA,\\
\infty & \mbox{otherwise}.
\end{cases}
\]
Note that, for every $Z\in\cL'$, we have $\sigma_\cA(Z)=(-\ind_\cA)^\bullet(Z)=-\ind_\cA^\ast(-Z)$.


\subsection{Financial systems and systemic risk}

We consider a one-period economy in which uncertainty at the terminal date is modeled by the probability space $(\Omega,\cF,\probp)$. In this economy, we assume the existence of a {\em financial system} consisting of $d$ member institutions (for completeness we also allow for the case $d=1$). The possible terminal capital positions, i.e.\ assets net of liabilities, of these $d$ institutions belong to an admissible space
\[
\cX = \cX_1\times\cdots\times\cX_d \subset L^0(\R^d).
\]
For every $X=(X_1,\dots,X_d)\in\cX$ the random variables $X_1,\dots,X_d$ correspond to the capital positions of the various member institutions. Since $\cX$ contains all bounded random vectors, the space $\R^d$ can be naturally viewed as a linear subspace of $\cX$. We denote by $e$ the constant random vector with all components equal to $1$, i.e.
\[
e := (1,\dots,1) \in \R^d.
\]

\smallskip

The impact of the financial system on systemic risk is measured through an {\em impact map}
\[
\agg : \cX \to \cE
\]
where $\cE$ is a suitable admissible subspace of $L^0$. Hence, for every $X\in\cX$, the random variable $\agg(X)$ is interpreted as an indicator of the systemic risk posed by $X$; see Remark~\ref{ex: Lambda}.

\begin{definition}
We say that $\agg$ is {\em admissible} if it satisfies the following five properties:
\begin{enumerate}
\item[($\textrm{S}$1)] {\em Discrimination}: $\agg$ is not constant;
\item[($\textrm{S}$2)] {\em Normalization}: $\agg(0)=0$;
\item[($\textrm{S}$3)] {\em Monotonicity}: $\agg(X)\geq\agg(Y)$ for all $X,Y\in\cX$ such that $X\geq Y$;
\item[($\textrm{S}$4)] {\em Concavity}: $\agg(\lambda X+(1-\lambda)Y)\geq\lambda\agg(X)+(1-\lambda)\agg(Y)$ for all $X,Y\in\cX$ and $\lambda\in[0,1]$;
\item[($\textrm{S}$5)] {\em Semicontinuity}: The map $X\mapsto\E[\agg(X)W]$ is $\sigma(\cX,\cX')$-upper semicontinuous for every $W\in\cE'_+$.
\end{enumerate}
\end{definition}

\smallskip

The next proposition provides a number of sufficient conditions for the technical assumption (S5) to hold. Recall that the lattice operations on $\cX$ are performed component by component. Here, we use the standard notation for the limit superior of a sequence of random variables.

\begin{definition}
We say that $\agg$ has the {\em Fatou property} if for every sequence $(X^n)\subset\cX$ and every $X\in\cX$
\[
X^n\to X \ \mbox{a.s.}, \ \sup_{n\in\N}|X^n|\in\cX \ \implies \ \agg(X) \geq \limsup_{n\to\infty}\agg(X^n).
\]
We say that $\agg$ is {\em surplus invariant} if $\agg(X)=\agg(\min(X,0))$ for every $X\in\cX$.
\end{definition}

\smallskip

\begin{proposition}
\label{prop: conditions for S5}
Assume that (S3) and (S4) hold. Then, (S5) holds in any of the following cases:
\begin{enumerate}[(i)]
  \item $\cX'_i$ is the norm dual of $\cX_i$ for every $i\in\{1,\dots,d\}$.
  \item $\cX_i=L^{\Phi_i}$ with $\Phi_i^\ast$ being $\Delta_2$ (e.g.\ $\cX_i=L^\infty$) for every $i\in\{1,\dots,d\}$ and $\agg$ has the Fatou property.
  \item $\agg$ is surplus invariant and has the Fatou property.
\end{enumerate}
\end{proposition}
\begin{proof}
Throughout the proof fix $W\in\cE'_+$ and define a functional $\varphi_W:\cX\to\R$ by setting
\[
\varphi_W(X) := \E[\agg(X)W].
\]
Note that $\varphi_W$ is concave and nondecreasing by (S3) and (S4). Assume that {\em (i)} holds. In this case, we can apply the Extended Namioka-Klee Theorem from Biagini and Frittelli~\cite{BiaginiFrittelli2009} to infer that $\varphi_W$ is upper semicontinuous (in fact, continuous) with respect to the norm topology on $\cX$. As the space $\cX'$ coincides with the norm dual of $\cX$ by assumption, it follows from Corollary 5.99 in Aliprantis and Border~\cite{AliprantisBorder2006} that $\varphi_W$ is also $\sigma(\cX,\cX')$-upper semicontinuous.
\smallskip

We make some preliminary observations before proceeding with the proof of \textit{(ii)} and \textit{(iii)}. First, we note that the Fatou property of $S$ implies that $\varphi_W$ is sequentially upper semicontinuous with respect to order convergence, i.e.\ dominated almost-sure convergence. Indeed, consider a sequence $(X^n)\subset\cX$ that converges almost surely to some $X\in\cX$ and such that $\sup_{n\in\N}|X^n|\leq M$ for some $M\in\cX$. Since $|\agg(X^n)|\leq\max(|\agg(M)|,|\agg(-M)|)$ for every $n\in\N$ by (S3), it follows from the Fatou property of $\agg$ and from the Fatou Lemma that
\[
\varphi_W(X) \geq \E\Big[\limsup_{n\to\infty}\agg(X^n)W\Big] \geq \limsup_{n\to\infty}\E[\agg(X^n)W] = \limsup_{n\to\infty}\varphi_W(X^n),
\]
as claimed. Second, Theorem 2.6.4 in Meyer-Nieberg~\cite{MeyerNieberg1991} tells us that, for every $i\in\{1,\dots,d\}$, the order-continuous dual of $\cX_i$, i.e.\ the space of linear functionals that are continuous with respect to order convergence, coincides with $\cX'_i$. This implies that the order-continuous dual of $\cX$ also coincides with $\cX'$. Denote by $\cX^\sim_n$ the order-continuous dual of $\cX$. We establish (S5) by showing that the upper semicontinuity of $\varphi_W$ with respect to order convergence implies its $\sigma(\cX,\cX^\sim_n)$-upper semicontinuity.

\smallskip

Assume that {\em (ii)} holds. If $d=1$, the desired assertion follows from Theorem 4.4 in Delbaen and Owari~\cite{DelbaenOwari2019} (see also Theorem 3.2 in Delbaen~\cite{Delbaen2002} for the bounded case and Theorem 3.7 in Gao et al.~\cite{GaoLeungXanthos2019} for the Orlicz case in a nonatomic setting). This result can be extended to a multivariate setting by using the results in Leung and Tantrawan~\cite{LeungTantrawan2018}. We use their notation and terminology. Observe first that the constant vector $e$ is a strictly-positive element in $\cX^\sim_n$. Second, note that all the spaces $\cX_i$'s are monotonically complete by Theorem 2.4.22 in~\cite{MeyerNieberg1991}, admit a special modular by Example 3.1 in~\cite{LeungTantrawan2018}, and their norm duals are order continuous by Remark 3.5 in~\cite{DelbaenOwari2019}. This implies that $\cX$ is also monotonically complete, admits a special modular, and its norm dual is order continuous. As a result, we can apply Theorem 3.4 in~\cite{LeungTantrawan2018} to conclude that $\cX$ satisfies property (P1) of that paper. This property implies that every concave functional on $\cX$ that is upper semicontinuous with respect to order convergence, as our $\varphi_W$, is automatically $\sigma(\cX,\cX^\sim_n)$-upper semicontinuous. This delivers the desired result.

\smallskip

Finally, assume that {\em (iii)} holds. In this case, the functional $\varphi_W$ is surplus invariant in the sense of Gao and Munari~\cite{GaoMunari2017}. Since $\varphi_W$ is concave and upper semicontinuous with respect to order convergence, we can apply Theorem 21 in~\cite{GaoMunari2017} to infer that $\varphi_W$ is $\sigma(\cX,\cX^\sim_n)$-upper semicontinuous. As above, this delivers the desired result.
\end{proof}

\smallskip

\begin{remark}
\label{ex: Lambda}
(i) In the literature, the impact map is typically derived from an {\em aggregation function}
\[
\Lambda:\R^d\to\R
\]
by setting $\agg(X)=\Lambda(X)$ for every $X\in\cX$. We refer to the literature cited in the introduction for a discussion of concrete examples. Clearly, the choice of $\Lambda$ limits the choice of the space $\cE$ since, for instance, one needs to ensure that the random variables $\Lambda(X)$'s are integrable. This is typically done either by working in a space of bounded positions or by working in an Orlicz space where the Orlicz functions are defined in terms of $\Lambda$. To avoid having to worry about this aspect, we have defined the impact map as a map between abstract spaces.

\smallskip

(ii) If $\Lambda$ is assumed to be nonconstant, nondecreasing, concave, and to satisfy $\Lambda(0)=0$, then the corresponding $\agg$ clearly fulfills properties (S1)-(S4). Moreover, as $\Lambda$ is automatically continuous by concavity, $\agg$ has automatically the Fatou property. Hence, we can use Proposition~\ref{prop: conditions for S5} to ensure property (S5).
\end{remark}

\smallskip

We now assume that the regulator has defined acceptable levels of systemic risk by specifying a set
\[
\cA\subset\cE
\]
called the {\em acceptance set}: The financial system with capital positions $X\in\cX$ is deemed to have an acceptable level of systemic risk if the  systemic risk indicator $\agg(X)$ belongs to $\cA$.

\begin{definition}
We say that $\cA$ is {\em admissible} if it satisfies the following properties:
\begin{enumerate}
  \item[($\textrm{A}$1)] {\em Discrimination}: $\agg^{-1}(\cA)$ is a nonempty proper subset of $\cX$;
  \item[($\textrm{A}$2)] {\em Normalization}: $0\in\cA$;
  \item[($\textrm{A}$3)] {\em Monotonicity}: $\cA+\cE_+\subset\cA$;
  \item[($\textrm{A}$4)] {\em Convexity}: $\lambda\cA+(1-\lambda)\cA\subset\cA$ for every $\lambda\in[0,1]$;
  \item[($\textrm{A}$5)] {\em Closedness}: $\cA$ is $\sigma(\cE,\cE')$-closed.
\end{enumerate}
\end{definition}

\smallskip

The next proposition highlights a variety of situations where assumption (A5) is always satisfied.

\begin{definition}
We say that $\cA$ is {\em Fatou closed} if for every sequence $(U_n)\subset\cA$ and every $U\in\cE$
\[
U_n\to U \ \mbox{a.s.}, \ \sup_{n\in\N}|U_n|\in\cE \ \implies \ U\in\cA.
\]
We say that $\cA$ is {\em law invariant} if for every $U\in\cA$ and every $V\in\cE$ with the same probability distribution as $U$ we have $V\in\cA$. Moreover, we say that $\cA$ is {\em surplus invariant} if for every $U\in\cA$ and every $V\in\cE$ such that $\min(V,0)=\min(U,0)$ we have $V\in\cA$.
\end{definition}

\smallskip

\begin{proposition}
Assume that (A3) and (A4) hold. Then, (A5) holds in any of the following cases:
\begin{enumerate}[(i)]
  \item $\cE'$ is the norm dual of $\cE$ and $\cA$ is norm closed.
  \item $\cE=L^\Phi$ with $\Phi^\ast$ being $\Delta_2$ (e.g.\ $\cE=L^\infty$) and $\cA$ is Fatou closed.
  \item $\cE=L^\Phi$ with $(\Omega,\cF,\probp)$ nonatomic and $\cA$ is law invariant and Fatou closed.
  \item $\cA$ is surplus invariant and Fatou closed.
\end{enumerate}
\end{proposition}
\begin{proof}
The desired assertion holds under {\em (i)} by Theorem 5.98 in Aliprantis and Border~\cite{AliprantisBorder2006}; under {\em (ii)} by Theorem 4.1 in Delbaen and Owari~\cite{DelbaenOwari2019} (see also Theorem 3.2 in Delbaen~\cite{Delbaen2002} in the bounded case and Theorem 3.7 in Gao et al.~\cite{GaoLeungXanthos2019} in the Orlicz case in a nonatomic setting); under {\em (iii)} by Corollary 4.6 in Gao et al.~\cite{GaoLeungMunariXanthos2018}; under {\em (iv)} by Theorem 8 in Gao and Munari~\cite{GaoMunari2017}.
\end{proof}


\section{``First allocate, then aggregate''-type systemic risk measures}

In this section we focus on systemic risk measures of ``first allocate, then aggregate'' type. After discussing some conditions for their representability, we establish a general dual representation and provide a detailed analysis of the properties of the corresponding systemic acceptance sets and ``penalty functions''. Throughout the section we fix an admissible impact map $\agg$ and an admissible acceptance set $\cA$.


\subsection{The systemic risk measure $\rho$}

``First allocate, then aggregate''-type systemic risk measures are defined as follows (we adopt the usual convention $\inf\emptyset=\infty$):

\begin{definition}
We define the map $\rho:\cX\to[-\infty,\infty]$ by setting
\begin{equation}
\label{eq:srm_def}
\rho(X) := \inf\left\{\sum_{i=1}^dm_i \,; \ m\in\R^d, \ \agg(X+m)\in\cA\right\}.
\end{equation}
\end{definition}

\smallskip

The above map determines the minimum amount of aggregate capital that can be allocated to the member institutions to ensure that the level of systemic risk of the financial system is acceptable. We start by observing that~\eqref{eq:srm_def} can be rewritten as
\begin{equation}
\label{eq: multi-asset representation}
\rho(X)=\inf\{\pi(m) \,; \ m\in\R^d, \ X+m\in\agg^{-1}(\cA)\}, \ \ \ \ \pi(m) = \sum_{i=1}^dm_i.
\end{equation}
As a result, $\rho$ belongs to the broad class of risk measures introduced in Frittelli and Scandolo~\cite{FrittelliScandolo2006} and thoroughly studied in Farkas et al.~\cite{kochmunari_multiasset}. We exploit this link in a systematic way. The first proposition collects some basic properties of the ``systemic acceptance set'' $\agg^{-1}(\cA)$ and of the risk measure $\rho$.

\begin{proposition}
\label{prop: properties rho}
\begin{enumerate}[(i)]
    \item The set $\agg^{-1}(\cA)$ is monotone, convex, $\sigma(\cX,\cX')$-closed, and contains $0$.
    \item The systemic risk measure $\rho$ is nonincreasing, convex, and satisfies $\rho(0)\leq0$. Moreover, $\rho$ satisfies the multivariate version of cash-additivity, i.e.
\[
\rho(X+m) = \rho(X)-\sum_{i=1}^dm_i
\]
for every $X\in\cX$ and every $m\in\R^d$.
\end{enumerate}
\end{proposition}
\begin{proof}
{\em (i)} It is straightforward to prove that $\agg^{-1}(\cA)$ contains $0$ and that it is monotone and convex. To show $\sigma(\cX,\cX')$-closedeness, it is enough to recall that $\cB(\cA)\subset\cE'_+$ and use~\eqref{eq: external characterization} to get
\[
\agg^{-1}(\cA) = \{X\in\cX \,; \ \agg(X)\in\cA\} = \{X\in\cX \,; \ \E[\agg(X)W]\geq\sigma_\cA(W), \ \forall W\in\cB(\cA)\}.
\]
The claim follows immediately from (S5).

\smallskip

{\em (ii)} The stated properties of $\rho$ are straightforward; see also Lemma~2 in Farkas et al.~\cite{kochmunari_multiasset}.
\end{proof}


\subsection{Properness and lower semicontinuity of $\rho$}

In order to admit a dual representation, the risk measure $\rho$ needs to be proper and lower semicontinuous. We highlight a number of sufficient conditions for this to be the case. We start with a simple characterization of properness provided we already know that $\rho$ is lower semicontinuous. Recall that, by definition, $\rho$ is proper if it never attains the value $-\infty$ and is finite at some point.

\begin{proposition}
\label{prop: properness rho}
If $\rho$ is $\sigma(\cX,\cX')$-lower semicontinuous, then $\rho$ is proper if and only if $\rho(0)>-\infty$.
\end{proposition}
\begin{proof}
We know that $\rho(0)<\infty$ by Proposition~\ref{prop: properties rho}. As a result, the above equivalence follows from the fact that a $\sigma(\cX,\cX')$-lower semicontinuous convex map that assumes the value $-\infty$ cannot assume any finite value; see e.g.\ Proposition 2.2.5 in Z\u{a}linescu~\cite{Zalinescu2002}.
\end{proof}

\smallskip

In contrast to the standard univariate (cash-additive) case, the closedeness of $\agg^{-1}(\cA)$ does not imply the lower semicontinuity of $\rho$; see Example 1 in Farkas et al.~\cite{kochmunari_multiasset}. The purpose of the next result is to provide a number of sufficient conditions for $\rho$ to be lower semicontinuous. The last two conditions are particularly easy to verify and often satisfied in the literature.

\begin{proposition}
\label{prop: lsc rho}
The following statements hold:
\begin{enumerate}[(i)]
\item Assume that $\Omega$ is finite. If $\rho(0)>-\infty$, then $\rho$ is finite valued and continuous.
\item Assume that $\cX'_i$ is the norm dual of $\cX_i$ for every $i\in\{1,\dots,d\}$. If $\rho$ is finite valued, then $\rho$ is $\sigma(\cX,\cX')$-lower semicontinuous.
\item Assume that $\cX'_i$ is the norm dual of $\cX_i$ for every $i\in\{1,\dots,d\}$. If $\agg^{-1}(\cA)$ has nonempty interior in the norm topology and $\rho(0)>-\infty$, then $\rho$ is finite valued and $\sigma(\cX,\cX')$-lower semicontinuous.
\item Set $\cM_0:=\big\{m\in\R^d \,; \ \sum_{i=1}^dm_i=0\big\}$. If $\agg^{-1}(\cA)\cap\cM_0=\{0\}$, then $\rho$ is proper and $\sigma(\cX,\cX')$-lower semicontinuous.
\item If $\cA\cap\R_-=\{0\}$ and $\agg(m)\in(-\infty,0)$ for every nonzero $m\in\cM_0$, then $\rho$ is proper and $\sigma(\cX,\cX')$-lower semicontinuous.
\end{enumerate}
\end{proposition}
\begin{proof}
{\em (i)} Since $\Omega$ is finite, $e$ is an interior point of $\cX_+$. Then, the desired result follows from Proposition 1 in Farkas et al.~\cite{kochmunari_multiasset}.

\smallskip

{\em (ii)} It follows from the Extended Namioka-Klee Theorem in Biagini and Frittelli~\cite{BiaginiFrittelli2009} that $\rho$ is lower semicontinuous (in fact, continuous) with respect to the norm topology. Then, $\rho$ is also lower semicontinuous with respect to $\sigma(\cX,\cX')$ by virtue of Corollary 5.99 in Aliprantis and Border~\cite{AliprantisBorder2006}.

\smallskip

{\em (iii)} Note that $e$ is a strictly-positive element of $\cX$, i.e.\ for every $Z\in\cX'_+\setminus\{0\}$ we have
\[
\E[\langle e,Z\rangle] = \sum_{i=1}^d\E[Z_i] > 0.
\]
Proposition 2 in~\cite{kochmunari_multiasset} implies that $\rho$ is finite valued so that {\em (ii)} can be applied.

\smallskip

{\em (iv)} For every $X\in\cX$ it is not difficult to show that
\[
\rho(X) = \inf\left\{r\in\R \,; \ X+\frac{r}{d}e\in\agg^{-1}(\cA)+\cM_0\right\};
\]
see Lemma 3 in~\cite{kochmunari_multiasset}. Then, it follows from Proposition~\ref{prop: properties rho} that
\[
\agg^{-1}(\cA)+\cM_0-\frac{r}{d}e \subset \{X\in\cX \,; \ \rho(X)\leq r\} \subset \cl\bigg(\agg^{-1}(\cA)+\cM_0-\frac{r}{d}e\bigg)
\]
for every $r\in\R$, where $\cl$ denotes the closure operator with respect to $\sigma(\cX,\cX')$. To establish the desired lower semicontinuity, we show that $\agg^{-1}(\cA)+\cM_0$ is $\sigma(\cX,\cX')$-closed. To this effect, recall from Proposition~\ref{prop: properties rho} that $\agg^{-1}(\cA)$ is convex and $\sigma(\cX,\cX')$-closed. Moreover, $\cM_0$ is a finite-dimensional vector space and $\agg^{-1}(\cA)\cap\cM_0=\{0\}$. The closedness criterion in Dieudonn\'{e}~\cite{Dieudonne1966} now implies that $\agg^{-1}(\cA)+\cM_0$ is $\sigma(\cX,\cX')$-closed. Properness follows from Proposition~\ref{prop: properness rho}.

\smallskip

{\em (v)} Let $m\in\cM_0$. By assumption, we have $\agg(m)\in\cA$ if and only if $m=0$. This yields $\agg^{-1}(\cA)\cap\cM_0=\{0\}$ and the desired statement immediately follows from point {\em (iv)}.
\end{proof}


\subsection{The dual representation of $\rho$}

We have already mentioned that, in view of~\eqref{eq: multi-asset representation}, the risk measure $\rho$ belongs to the class of risk measures studied in Farkas et al.~\cite{kochmunari_multiasset}. The general results established in that paper can be exploited to derive a dual representation for $\rho$. This also follows from the general dual representation in Frittelli and Scandolo~\cite{FrittelliScandolo2006}.

\begin{definition}
We denote by $\cC$ the convex subset of $\cX'_+$ defined by
\[
\cC := \{Z\in\cX'_+ \,; \ \E[Z_1]=\cdots=\E[Z_d]=1\}.
\]
\end{definition}

\smallskip

\begin{theorem}
\label{theo: dual representation}
If $\rho$ is proper and $\sigma(\cX,\cX')$-lower semicontinuous, then $\cB(\agg^{-1}(\cA))\cap\cC\neq\emptyset$ and
\[
\rho(X) = \sup_{Z\in\cC}\{\sigma_{\agg^{-1}(\cA)}(Z)-\E[\langle X,Z\rangle]\}
\]
for every $X\in\cX$. The supremum can be restricted to $\cC\cap\cX'_{++}$ provided that $\cB(\agg^{-1}(\cA))\cap\cC\cap\cX'_{++}\neq\emptyset$.
\end{theorem}
\begin{proof}
Note that the cost functional $\pi$ in equation \eqref{eq: multi-asset representation} is defined on $\R^d\subset\cX$. It is easy to see that, for every $Z\in\cX'$, the functional $X\mapsto\E[\langle X,Z\rangle]$ is a positive extension of $\pi$ to $\cX$ if and only if $Z$ belongs to $\cC$. Since $\rho$ is proper and $\sigma(\cX,\cX')$-lower semicontinuous, it follows from Proposition~6 in~\cite{kochmunari_multiasset} that the barrier cone of $\agg^{-1}(\cA)$ contains positive linear extensions of the cost functional $\pi$ to $\cX$, i.e.\ we have $\cB(\agg^{-1}(\cA))\cap\cC\neq\emptyset$. The desired representation is now a consequence of Theorem~3 in~\cite{kochmunari_multiasset}.

\smallskip

Now, assume we find $Z^\ast\in\cB(\agg^{-1}(\cA))\cap\cC\cap\cX'_{++}$ and take any element $Z\in\cB(\agg^{-1}(\cA))\cap\cC$. For every $X\in\cX$ and every $\lambda\in(0,1)$ we have $\lambda Z^\ast+(1-\lambda)Z\in\cC$ and
\[
\lambda(\sigma_{\agg^{-1}(\cA)}(Z^\ast)-\E[\langle X,Z^\ast\rangle])+
(1-\lambda)(\sigma_{\agg^{-1}(\cA)}(Z)-\E[\langle X,Z\rangle])
\leq
\sup_{Z'\in\cC\cap\cX'_{++}}\{\sigma_{\agg^{-1}(\cA)}(Z')-\E[\langle X,Z'\rangle]\}
\]
by concavity of $\sigma_{\agg^{-1}(\cA)}$. Letting $\lambda$ tend to $0$ and taking a supremum over $Z$ yields
\[
\rho(X) \leq \sup_{Z'\in\cC\cap\cX'_{++}}\{\sigma_{\agg^{-1}(\cA)}(Z')-\E[\langle X,Z'\rangle]\}.
\]
The converse inequality is clear. This establishes the last assertion and concludes the proof.
\end{proof}

\smallskip

\begin{remark}
\label{rem: fenchel moreau for rho}
(i) We highlight the link between the dual representation in Theorem~\ref{theo: dual representation} and the standard Fenchel-Moreau representation; see also Remark~17 in Farkas et al.~\cite{kochmunari_multiasset}. To see it, note that the map $-\sigma_{\agg^{-1}(\cA)}(-\cdot)+\ind_{\cC}(-\cdot)$ is convex and lower semicontinuous and, if $\rho$ is proper and $\sigma(\cX,\cX')$-lower semicontinuous, it satisfies
\[
\rho(X) = \sup_{Z\in\cX'}\{\E[\langle X,Z\rangle]+\sigma_{\agg^{-1}(\cA)}(-Z)-\ind_{\cC}(-Z)\}
\]
for every $X\in\cX$ by Theorem~\ref{theo: dual representation}. From the Fenchel-Moreau Theorem it follows that for every $Z\in\cX'$
\[
\rho^\ast(Z) = -\sigma_{\agg^{-1}(\cA)}(-Z)+\ind_{\cC}(-Z) =
\begin{cases}
\sup_{X\in\agg^{-1}(\cA)}\E[\langle X,Z\rangle] & \mbox{if $Z\in-\cC$},\\
\infty & \mbox{otherwise}.
\end{cases}
\]

\smallskip

(ii) The dual elements in $\cC$ can be naturally identified with $d$-dimensional vectors of probability measures on $(\Omega,\cF)$ that are absolutely continuous with respect to $\probp$ or, in case they have strictly-positive components, equivalent to $\probp$. This allows to reformulate the above dual representation in terms of probability measures. More concretely, denote by $\cQ(\probp)$, respectively $\cQ_e(\probp)$, the set of all $d$-dimensional vectors of probability measures over $(\Omega,\cF)$ that are absolutely continuous with respect to $\probp$, respectively equivalent to $\probp$. For every $\probq=(\probq_1,\dots,\probq_d)\in\cQ(\probp)$ and for every $X\in\cX$ we set
\[
\frac{d\probq}{d\probp} := \left(\frac{d\probq_i}{d\probp},\dots,\frac{d\probq_d}{d\probp}\right), \ \ \ \ \E_\probq[X] := \E\bigg[\bigg\langle X,\frac{d\probq}{d\probp}\bigg\rangle\bigg] = \sum_{i=1}^d\E_{\probq_i}[X_i].
\]
Moreover, for every $\probq\in\cQ(\probp)$ we define
\[
\sigma(\probq) := \sigma_{\agg^{-1}(\cA)}\bigg(\frac{d\probq}{d\probp}\bigg)=\inf_{X\in\agg^{-1}(\cA)}\E_\probq[X].
\]
If $\rho$ is proper and $\sigma(\cX,\cX')$-lower semicontinuous, then for every $X\in\cX$ we can write
\[
\rho(X) = \sup_{\probq\in\cQ(\probp),\,\frac{d\probq}{d\probp}\in\cX'}\{\sigma(\probq)-\E_\probq[X]\}.
\]
We can replace $\cQ(\probp)$ by $\cQ_e(\probp)$ in the above supremum provided that $\cB(\agg^{-1}(\cA))\cap\cC\cap\cX'_{++}\neq\emptyset$.
\end{remark}

\medskip

The condition $\cB(\agg^{-1}(\cA))\cap\cC\cap\cX'_{++}\neq\emptyset$ is necessary to be able to restrict the domain in the above dual representation to strictly-positive dual elements. In the terminology of convex analysis, this condition requires that the convex set $\agg^{-1}(\cA)$ admits a strictly-positive supporting functional that belongs to the special set $\cC$. In the next proposition we show that this always holds if the acceptance set $\cA$ is supported by a strictly-positive functional and the impact map $\agg$ is bounded above by a strictly-increasing affine function of the consolidated capital position.

\begin{proposition}
\label{prop: cond str pos}
Assume that $\cX_i=\cE$ for every $i\in\{1,\dots,d\}$. Moreover, suppose that $\cB(\cA)\cap\cE'_{++}\neq\emptyset$ and there exist $a\in(0,\infty)$ and $b\in\R$ such that
\[
\agg(X) \leq a\sum_{i=1}^dX_i+b
\]
for every $X\in\cX$. Then, $\cB(\agg^{-1}(\cA))\cap\cC\cap\cX'_{++}\neq\emptyset$.
\end{proposition}
\begin{proof}
Take $W\in\cB(\cA)\cap\cE'_{++}$ and set $Z=(aW,\dots,aW)\in\cC\cap\cX'_{++}$. Then, we easily see that
\[
\sigma_{\agg^{-1}(\cA)}(Z)
=
\inf_{X\in\agg^{-1}(\cA)}\E[\langle X,Z\rangle]
\geq
\inf_{X\in\agg^{-1}(\cA)}\E[(\agg(X)-b)W]
\geq
\sigma_\cA(W)-b\E[W]
>
-\infty.
\]
This delivers the desired assertion.
\end{proof}


\subsection{Characterizing the systemic acceptance set $\agg^{-1}(\cA)$}

Through the support function of the ``systemic acceptance set'' $\agg^{-1}(\cA)$, the dual representation of the systemic risk measure $\rho$ in Theorem~\ref{theo: dual representation} highlights the dependence on the two fundamental underlying ingredients: The impact map $\agg$ and the acceptance set $\cA$. The aim of this subsection is to provide a dual description of the systemic acceptance set by using ``penalty functions'' that are related to (but different from) the support function $\sigma_{\agg^{-1}(\cA)}$ and to investigate the main properties of these maps. Our analysis is based on the following definition.

\begin{definition}
\label{def: penalties}
We define two maps $\alpha,\alpha^+:\cX'\to[-\infty,+\infty]$ by setting
\[
\alpha(Z) := \sup_{W\in\cB(\cA)}\Big\{\sigma_\cA(W)+
\inf_{X\in\cX}\{\E[\langle X,Z\rangle]-\E[\agg(X)W]\}\Big\},
\]
\[
\alpha^+(Z) := \sup_{W\in\cB(\cA)\cap(\cE'_{++}\cup\{0\})}\Big\{\sigma_\cA(W)+
\inf_{X\in\cX}\{\E[\langle X,Z\rangle]-\E[\agg(X)W]\}\Big\}.
\]
\end{definition}

\smallskip

\begin{remark}
\label{rem: alpha K}
(i) It is easy to see that $\alpha$ and $\alpha^+$ are different in general. For example, if $d>1$ and $\cX_i=\cE$ for every $i\in\{1,\dots,d\}$ and we set $\agg(X) = \sum_{i=1}^dX_i$ for every $X\in\cX$ and $\cA=\cE_+$, then we have
\[
\alpha = -\delta_\cD \neq -\delta_{\cD\cap(\cX'_{++}\cup\{0\})} = \alpha^+
\]
where $\cD=\{Z\in\cX'_+ \,; \ Z_1=\cdots=Z_d\}$.

\smallskip

(ii) The above maps belong to the class of maps $\alpha_\cK:\cX'\to[-\infty,+\infty]$ defined by
\[
\alpha_\cK(Z) := \sup_{W\in\cK}\Big\{\sigma_\cA(W)+\inf_{X\in\cX}\{\E[\langle X,Z\rangle]-\E[\agg(X)W]\}\Big\},
\]
where $\cK$ is a convex cone in $\cB(\cA)$ such that $\lambda\cK+(1-\lambda)\cB(\cA)\subset\cK$ for every $\lambda\in[0,1]$. This will allow us to prove properties for $\alpha$ and $\alpha^+$ simultaneously. In fact, all properties of $\alpha$ and $\alpha^+$ we will consider are shared by the entire class.
\end{remark}

\smallskip

The next theorem records the announced dual representation of the systemic acceptance set and shows why the above maps are natural ``penalty functions''.

\begin{theorem}
\label{prop:Lambda-1(A)_dualrepr}
The systemic acceptance set $\agg^{-1}(\cA)$ can be represented as
\[
\agg^{-1}(\cA) = \bigcap_{Z\in\cX'}\{X\in\cX \,; \ \E[\langle X,Z\rangle]\geq\alpha(Z)\}.
\]
If $\cB(\cA)\cap\cE'_{++}\neq\emptyset$, then $\agg^{-1}(\cA)$ can also be represented as
\[
\agg^{-1}(\cA) = \bigcap_{Z\in\cX'}\{X\in\cX \,; \ \E[\langle X,Z\rangle]\geq\alpha^+(Z)\}.
\]
\end{theorem}
\begin{proof}
Let $\cK\subset\cB(\cA)$ be a convex cone as in Remark~\ref{rem: alpha K}. Note that, by concavity of $\sigma_\cA$, we can equivalently rewrite the representation~\eqref{eq: external characterization} applied to $\cA$ as
\[
\cA = \bigcap_{W\in\cK}\{U\in\cE \,; \ \E[UW]\geq\sigma_\cA(W)\}.
\]
Now, for each $W\in\cK\subset\cE'_+$ we consider the functional $\varphi_W:\cX\to\R$ defined by
\[
\varphi_W(X):=\E[\agg(X)W].
\]
As remarked in the proof of Proposition~\ref{prop: conditions for S5}, the functional $\varphi_W$ is concave by (S3) and (S4) and $\sigma(\cX,\cX')$-upper semicontinuous by (S5). Hence, it follows from the Fenchel-Moreau Theorem that
\[
\varphi_W(X) = \inf_{Z\in\cX'}\{\E[\langle X,Z\rangle]-(\varphi_W)^\bullet(Z)\}
\]
for every $X\in\cX$. As a result, we obtain
\begin{eqnarray*}
\agg^{-1}(\cA)
&=&
\{X\in\cX \,; \ \agg(X)\in\cA\} \\
&=&
\{X\in\cX \,; \ \E[\agg(X)W]\geq\sigma_\cA(W), \ \forall W\in\cK\} \\
&=&
\{X\in\cX \,; \ \E[\langle X,Z\rangle]-(\varphi_W)^\bullet(Z)\geq\sigma_\cA(W), \ \forall W\in\cK, \ \forall Z\in\cX'\} \\
&=&
\bigcap_{Z\in\cX'}\Big\{X\in\cX \,; \ \E[\langle X,Z\rangle] \geq\sup_{W\in\cK}\{\sigma_\cA(W)+(\varphi_W)^\bullet(Z)\}\Big\} \\
&=&
\bigcap_{Z\in\cX'}\{X\in\cX \,; \ \E[\langle X,Z\rangle] \geq\alpha_\cK(Z)\}.
\end{eqnarray*}
This delivers the desired representation when applied to $\cK=\cB(\cA)$ and $\cK=\cB(\cA)\cap(\cE'_{++}\cup\{0\})$.
\end{proof}

\smallskip

The next proposition collects some properties of the maps $\alpha$ and $\alpha^+$ and shows the relation between them. Here, we denote by $\dom(\alpha)$ the domain of finiteness of $\alpha$ (similarly for $\alpha^+$). In addition, we denote by $\cl$ the closure operator with respect to the topology $\sigma(\cX',\cX)$.

\begin{proposition}
\label{prop:properties_alpha}
The maps $\alpha,\alpha^+:\cX'\to[-\infty,\infty]$ satisfy the following properties (the statements about $\alpha^+$ require that $\cB(\cA)\cap\cE'_{++}\neq\emptyset$):
\begin{enumerate}[(i)]
    \item $\alpha$ and $\alpha^+$ take values in the interval $[-\infty,0]$.
    \item $\alpha$ and $\alpha^+$ are concave and positively homogeneous.
    \item $\alpha^+\leq\alpha$ with equality on $\dom(\alpha^+)$.
    \item $\dom(\alpha^+)\subset\dom(\alpha)\subset\cl(\dom(\alpha^+))\subset\cX'_+$.
\end{enumerate}
\end{proposition}
\begin{proof}
Throughout the proof we fix a convex cone $\cK\subset\cB(\cA)$ as in Remark~\ref{rem: alpha K}. The desired assertions will follow by taking $\cK=\cB(\cA)$ and $\cK=\cB(\cA)\cap(\cE'_{++}\cup\{0\})$.

\smallskip

{\em (i)} The representation of the systemic acceptance set $\agg^{-1}(\cA)$ established in the proof of Theorem~\ref{prop:Lambda-1(A)_dualrepr} yields $\alpha_\cK(Z)\leq\sigma_{\agg^{-1}(\cA)}(Z)\leq0$ for every $Z\in\cX'$.

\smallskip

{\em (ii)} To show that $\alpha_\cK$ is concave, set for all $Z\in\cX'$ and $W\in\cE'$
\[
\Phi(Z,W):=\inf_{X\in\cX}\{\sigma_\cA(W)+\E[\langle X,Z\rangle]-\E[\agg(X)W]\}.
\]
Being the infimum over the parameter $X$ of a function that is clearly jointly concave in $Z$ and $W$, we see that $\Phi$ is itself jointly concave. Since
\[
\alpha_\cK(Z) = \sup_{W\in\cK}\Phi(Z,W)
\]
for every $Z\in\cX'$, we infer that $\alpha_\cK$ is concave. To show that $\alpha_\cK$ is positively homogeneous, note first that $0$ always belongs to $\cK$, so that $\alpha_\cK(0)\geq0$. Together with point {\em (i)}, this implies that $\alpha_\cK(0)=0$. Finally, for $Z\in\cX'$ and $\lambda\in(0,\infty)$ we have
\begin{eqnarray*}
\alpha_\cK(\lambda Z)
&=&
\sup_{W\in\cK}\Big\{\sigma_\cA(W)+\inf_{X\in\cX}\{\lambda\E[\langle X,Z\rangle] -\E[\agg(X)W]\}\Big\} \\
&=&
\lambda\sup_{W\in\cK}\left\{\sigma_\cA\left(\frac{1}{\lambda}W\right)+\inf_{X\in\cX}\left\{ \E[\langle X,Z\rangle]-\E\left[\agg(X)\frac{1}{\lambda}W\right]\right\}\right\} \\
&=&
\lambda\sup_{W\in\cK}\Big\{\sigma_\cA(W)+\inf_{X\in\cX}\{\E[\langle X,Z\rangle] -\E[\agg(X)W]\}\Big\} \\
&=&
\lambda\alpha_\cK(Z),
\end{eqnarray*}
where we used that $\cK$ is a cone. This shows that $\alpha_\cK$ is positively homogeneous.

\smallskip

{\em (iii)} It is clear that $\alpha^+\leq\alpha$. To show that $\alpha^+=\alpha$ on $\dom(\alpha^+)$, take $Z\in\dom(\alpha^+)$ and note that
\[
\alpha(Z) = \sup_{W\in\cB(\cA)}\Phi(Z,W), \ \ \ \ \alpha^+(Z) = \sup_{W\in\cB(\cA)\cap\cE'_{++}}\Phi(Z,W).
\]
Take $W^\ast\in\cB(\cA)\cap\cE'_{++}$ such that $\Phi(Z,W^\ast)$ is finite. For each $W\in\cB(\cA)$ set $W_\lambda=\lambda W+(1-\lambda)W^\ast$ for $\lambda\in[0,1)$. Note that $(W_\lambda)\subset\cB(\cA)\cap\cE'_{++}$, so that
\[
\alpha^+(Z) \geq \Phi(Z,W_\lambda) \geq \lambda\Phi(Z,W)+(1-\lambda)\Phi(Z,W^\ast) \xrightarrow{\lambda\uparrow1}\Phi(Z,W).
\]
Taking a supremum over $W$ delivers $\alpha^+(Z)\geq\alpha(Z)$.

\smallskip

{\em (iv)} Note that $\dom(\alpha^+)\subset\dom(\alpha)$ by point {\em (iii)}. Since $\alpha\leq\sigma_{\agg^{-1}(\cA)}$ as proved in point {\em (i)}, we also have $\dom(\alpha)\subset\cB(\agg^{-1}(\cA))\subset\cX'_+$. As $\cX'_+$ is $\sigma(\cX',\cX)$-closed, it remains to show that $\dom(\alpha)\subset\cl(\dom(\alpha^+))$. To this effect, let $Z\in\dom(\alpha)$ and note that $\Phi(Z,W)$ must be finite for some $W\in\cB(\cA)$. Take $Z^\ast\in\dom(\alpha^+)$ and $W^\ast\in\cB(\cA)\cap\cE'_{++}$ such that $\Phi(Z^\ast,W^\ast)$ is finite. Then, for every $\lambda\in[0,1]$ we have
\[
\alpha^+(\lambda Z+(1-\lambda)Z^\ast) \geq \Phi(\lambda Z+(1-\lambda)Z^\ast,\lambda W+(1-\lambda)W^\ast) \geq \lambda\Phi(Z,W)+(1-\lambda)\Phi(Z^\ast,W^\ast) > -\infty
\]
by the joint convexity of $\Phi$. The claim follows by letting $\lambda\uparrow1$.
\end{proof}


\subsubsection*{The case where $\agg$ is induced by $\Lambda$}

As mentioned in Remark~\ref{ex: Lambda}, the bulk of the literature has focused on the case where the impact function is based on an aggregation function $\Lambda:\R^d\to\R$. The last part of this subsection is devoted to provide an equivalent formulation of $\alpha$ and $\alpha^+$ in this situation. We focus on the positive cone $\cX'_+$ because both maps take nonfinite values elsewhere. For ease of notation, for every $Z\in\cX'_+$ we set
\[
E_+(Z) := \bigcup_{i=1}^d\{Z_i>0\} \in \cF.
\]

\smallskip

\begin{proposition}
\label{prop: alpha in the Lambda case}
Assume that $\cX$ is closed with respect to multiplications by characteristic functions, i.e.\ for every $X\in\cX$ and $E\in\cF$ we have $(\one_EX_1,\dots,\one_EX_d)\in\cX$. Moreover, consider a nonconstant, nondecreasing, concave function $\Lambda:\R^d\to\R$ satisfying $\Lambda(0)=0$ and assume that $\agg(X) = \Lambda(X)$ for every $X\in\cX$. Then, the following statements hold for every nonzero $Z\in\cX'_+$:
\begin{enumerate}[(i)]
  \item We have $\cB(\cA)\cap\{W\in\cE'_+ \,; \ \mbox{$W>0$ on $E_+(Z)$}\}\neq\emptyset$ and
\[
\alpha(Z) = \sup_{W\in\cB(\cA),\,W>0 \,\mbox{\footnotesize on}\, E_+(Z)}
\bigg\{\sigma_\cA(W)+\E\bigg[\one_{\{W>0\}}\Lambda^\bullet\bigg(\frac{Z}{W}\bigg)W\bigg]\bigg\}.
\]
  \item If $\cB(\cA)\cap\cE'_{++}\neq\emptyset$, then
\[
\alpha^+(Z) = \sup_{W\in\cB(\cA)\cap\cE'_{++}}\bigg\{\sigma_\cA(W)+
\E\bigg[\Lambda^\bullet\bigg(\frac{Z}{W}\bigg)W\bigg]\bigg\}
\]
\end{enumerate}
In both cases, the ratio $\frac{Z}{W}$ is understood component by component.
\end{proposition}
\begin{proof}
Let $\cK\subset\cB(\cA)$ be a convex cone as in Remark~\ref{rem: alpha K} and fix a nonzero element $Z\in\cX'_+$. Inspired by Ararat and Rudloff~\cite{AraratRudloff}, we invoke Theorem 14.60 in Rockafellar and Wets~\cite{RockafellarWets2009} to get
\begin{equation}
\label{eq: auxiliary Lambda 0}
\inf_{X\in\cX}\{\E[\langle X,Z\rangle]-\E[\Lambda(X)W]\} = \E\Big[\inf_{x\in\R^d}\{xZ-\Lambda(x)W\}\Big]
\end{equation}
for every $W\in\cK$ (this result requires that $\cX$ be closed with respect to multiplications by characteristic functions). Recall that $\cK\subset\cE'_+$ and note that for every $W\in\cK$ we have
\begin{equation}
\label{eq: auxiliary Lambda}
\inf_{x\in\R^d}\{xZ-\Lambda(x)W\}=
\begin{cases}
\Lambda^\bullet\big(\frac{Z}{W}\big)W & \mbox{on} \ \{W>0\},\\
0 & \mbox{on} \ \{W=0\}\cap(E_+(Z))^c,\\
-\infty & \mbox{on} \ \{W=0\}\cap E_+(Z).
\end{cases}
\end{equation}
It follows from the definition of $\alpha_\cK$ and \eqref{eq: auxiliary Lambda 0} that
\[
\alpha_\cK(Z) = \sup_{W\in\cK}\Big\{\sigma_\cA(W)+\E\Big[\inf_{x\in\R^d}\{xZ-\Lambda(x)W\}\Big]\Big\}.
\]
Clearly, no $W\in\cB(\cA)$ with $\probp(\{W=0\}\cap E_+(Z))>0$ contributes to the above supremum, so that
\begin{eqnarray*}
\alpha_\cK(Z)
&=&
\sup_{W\in\cK,\,W>0 \,\mbox{\em\footnotesize on}\, E_+(Z)}\Big\{\sigma_\cA(W)+\E\Big[\inf_{x\in\R^d}\{xZ-\Lambda(x)W\}\Big]\Big\} \\
&=&
\sup_{W\in\cK,\,W>0 \,\mbox{\em\footnotesize on}\, E_+(Z)}\Big\{\sigma_\cA(W)+\E\Big[\one_{\{W>0\}}\inf_{x\in\R^d}\{xZ-\Lambda(x)W\}\Big]\Big\} \\
&=&
\sup_{W\in\cK,\,W>0 \,\mbox{\footnotesize on}\, E_+(Z)}
\Big\{\sigma_\cA(W)+\E\Big[\one_{\{W>0\}}\Lambda^\bullet\big(\tfrac{Z}{W}\big)W\Big]\Big\},
\end{eqnarray*}
where we used~\eqref{eq: auxiliary Lambda} in the last equality. The desired assertions follow by taking $\cK=\cB(\cA)$ and $\cK=\cB(\cA)\cap(\cE'_{++}\cup\{0\})$.
\end{proof}


\subsection{Characterizing the support function $\sigma_{\agg^{-1}(\cA)}$}

As we have already noticed, the dual representation in Theorem~\ref{theo: dual representation} depends on the impact map $\agg$ and the acceptance set $\cA$ through the support function of the systemic acceptance set $\agg^{-1}(\cA)$. The goal of this subsection is to provide an equivalent description of the support function that relies on the ``penalty functions'' $\alpha$ and $\alpha^+$. This is a direct consequence of the results in the preceding subsection. Here, we denote by $\ucl(\alpha)$ the $\sigma(\cX',\cX)$-upper semicontinuous hull of $\alpha$, i.e.\ the smallest $\sigma(\cX',\cX)$-upper semicontinuous map dominating $\alpha$ (similarly for $\alpha^+$).

\begin{theorem}
\label{theo: representation support counterimage}
The support function $\sigma_{\agg^{-1}(\cA)}$ can be represented as
\[
\sigma_{\agg^{-1}(\cA)}=\ucl(\alpha).
\]
If $\cB(\cA)\cap\cE'_{++}\neq\emptyset$, then $\sigma_{\agg^{-1}(\cA)}$ can also be represented as
\[
\sigma_{\agg^{-1}(\cA)}=\ucl(\alpha^+).
\]
\end{theorem}
\begin{proof}
Let $\cK\subset\cB(\cA)$ be a convex cone as in Remark~\ref{rem: alpha K}. In the proof of Theorem~\ref{prop:Lambda-1(A)_dualrepr} we established that
\begin{equation}
\label{eq: auxiliary representation with alpha K}
\agg^{-1}(\cA) = \bigcap_{Z\in\cX'}\{X\in\cX \,; \ \E[\langle X,Z\rangle]\geq\alpha_\cK(Z)\}.
\end{equation}
This implies that $\alpha_\cK\leq\sigma_{\agg^{-1}(\cA)}$. It follows from the $\sigma(\cX',\cX)$-upper semicontinuity of $\sigma_{\agg^{-1}(\cA)}$ that we also have $\ucl(\alpha_\cK)\leq\sigma_{\agg^{-1}(\cA)}$. In particular, $\ucl(\alpha_\cK)$ never takes the value $\infty$. Moreover, note that $\ucl(\alpha_\cK)(0)\geq\alpha_\cK(0)=0$. As a result, Proposition 2.2.7 in Z\u{a}linescu~\cite{Zalinescu2002} tells us that $\ucl(\alpha_\cK)$ inherits concavity and positive homogeneity from $\alpha_\cK$. Note that $\alpha_\cK$ can be replaced by $\ucl(\alpha_\cK)$ in~\eqref{eq: auxiliary representation with alpha K}. Since the only $\sigma(\cX',\cX)$-upper semicontinuous map $\sigma:\cX'\to[-\infty,\infty)$ that is concave and positively homogeneous and satisfies
\[
\agg^{-1}(\cA) = \bigcap_{Z\in\cX'}\{X\in\cX \,; \ \E[\langle X,Z\rangle]\geq\sigma(Z)\}
\]
is precisely the support function of $\agg^{-1}(\cA)$, see e.g.\ Theorem 7.51 in Aliprantis and Border~\cite{AliprantisBorder2006}, we conclude that $\ucl(\alpha_\cK)=\sigma_{\agg^{-1}(\cA)}$ must hold. The desired assertions follow by taking $\cK=\cB(\cA)$ and $\cK=\cB(\cA)\cap(\cE'_{++}\cup\{0\})$.
\end{proof}

\smallskip

It is natural to ask whether taking the upper semicontinuous hull in Theorem~\ref{theo: representation support counterimage} is redundant in the sense that $\alpha$ and/or $\alpha^+$ are upper semicontinuous in the first place and, hence, coincide with the support function $\sigma_{\agg^{-1}(\cA)}$. As illustrated by the following example, the answer is negative in general.

\begin{example}
\label{ex:alpha_not_usc}
Let $(\Omega,\cF,\probp)$ be nonatomic and consider the pairs given by $(\cX,\cX')=(L^\infty(\R^d),L^1(\R^d))$ and $(\cE,\cE')=(L^\infty(\R),L^1(\R))$. Fix $\lambda\in(0,1)$ and for every $U\in L^0(\R)$ define the {\em Value at Risk} and {\em Expected Shortfall} of $U$ at level $\lambda$ by
\[
\VaR_\lambda(U) := \inf\{m\in\R \,; \ \probp(U+m<0)\leq\lambda\}, \ \ \ \ \ES_\lambda(U) := \frac{1}{\lambda}\int_0^\lambda\VaR_\mu(U)\,d\mu.
\]
Define $\agg:\cX\to\cE$ and $\cA\subset\cE$ by setting
\[
S(X) = \sum_{i=1}^d\min(X_i,0), \ \ \ \ \cA = \{U\in\cE \,; \ \ES_\lambda(U)\leq0\}.
\]
It is immediate to see that $\agg^{-1}(\cA)=\cX_+$, so that
\[
\sigma_{\agg^{-1}(\cA)} = -\ind_{\cX'_+} = -\ind_{L^1_+(\R^d)}.
\]
To determine $\alpha$, take any $Z\in\cX'_+$ and recall from Theorem 4.52 in F\"{o}llmer and Schied~\cite{FoellmerSchied2016} that
\[
\sigma_\cA=-\ind_{\cB(\cA)}, \ \ \ \cB(\cA)=\left\{W\in\cE'_+ \,; \ W\leq\frac{1}{\lambda}\E[W]\right\}.
\]
As a result, we infer that
\[
\alpha(Z)
=
\sup_{W\in\cE'_+,\,W\leq\frac{\E[W]}{\lambda}}\inf_{X\in\cX}\E\left[\sum_{i=1}^d (X_iZ_i-\min(X_i,0)W)\right]
=
\sup_{W\in\cE'_+,\,W\leq\frac{\E[W]}{\lambda}}\inf_{X\in\cX_+}\E\left[\sum_{i=1}^d X_i(W-Z_i)\right].
\]
Now, if $Z_j$ is not bounded for some $j\in\{1,\dots,d\}$, then $\probp(W-Z_j<0)>0$ for every $W\in\cB(\cA)$ and
\[
\inf_{X\in\cX_+}\E\left[\sum_{i=1}^d X_i(W-Z_i)\right] \leq \inf_{n\in\N}\E[n\mathbbm{1}_{\{W-Z_j<0\}}(W-Z_j)] = -\infty.
\]
In this case, we have $\alpha(Z)=-\infty$. Otherwise, if $Z$ is bounded, set $W=\max_{i\in\{1,\dots,d\}}\|Z_i\|_\infty\in\cB(\cA)$ and observe that
\[
0 \geq \alpha(Z) \geq \inf_{X\in\cX_+}\E\left[\sum_{i=1}^d X_i(W-Z_i)\right] = 0.
\]
In conclusion, we have
\[
\alpha = -\ind_{\cX_+} = -\ind_{L^\infty_+(\R^d)}.
\]
Since $L^\infty(\R)\neq L^1(\R)$ when the underlying probability space is nonatomic, we conclude that $\sigma_{\agg^{-1}(\cA)}\neq\alpha$. The same conclusion holds for $\alpha^+$ as well (note that $\cB(\cA)\cap\cE'_{++}\neq\emptyset$). This follows because, by Proposition~\ref{prop:properties_alpha}, we always have $\alpha^+\leq\alpha$ . Alternatively, we can repeat the above argument and find that $\alpha^+=\alpha$ in our situation.
\end{example}

\smallskip

\begin{remark}
By combining the dual representation in Theorem~\ref{theo: dual representation} and the representation of $\sigma_{\agg^{-1}(\cA)}$ obtained in Theorem~\ref{theo: representation support counterimage}, we see that
\begin{equation}
\label{eq: dual representation usc}
\rho(X) = \sup_{Z\in\cC}\{\ucl(\alpha)(Z)-\E[\langle X,Z\rangle]\} = \sup_{Z\in\cX'}\{\ucl(\alpha)(Z)-\ind_{\cC}(Z)-\E[\langle X,Z\rangle]\}
\end{equation}
for every $X\in\cX$. If the equality $\sigma_{\agg^{-1}(\cA)}=\alpha$ holds, then we can drop the upper-semicontinuous hull in the representation~\eqref{eq: dual representation usc} and obtain
\begin{equation}
\label{eq: dual representation usc simplified}
\rho(X) = \sup_{Z\in\cC}\{\alpha(Z)-\E[\langle X,Z\rangle]\} = \sup_{Z\in\cX'}\{\alpha(Z)-\ind_{\cC}(Z)-\E[\langle X,Z\rangle]\}
\end{equation}
for every $X\in\cX$. One may wonder whether the ``simplified'' representation~\eqref{eq: dual representation usc simplified} holds even if the equality $\sigma_{\agg^{-1}(\cA)}=\alpha$ does not hold. Note that $\ucl(\alpha)-\ind_{\cC}$ is concave and $\sigma(\cX',\cX)$-upper semicontinuous and that $\alpha-\ind_{\cC}$ is concave. As a result, the ``simplified'' representation holds if and only if
\[
\ucl(\alpha-\ind_{\cC}) = \ucl(\alpha)-\ind_{\cC}.
\]
The same holds with $\alpha^+$ instead of $\alpha$ (provided that $\cB(\cA)\cap\cE'_{++}\neq\emptyset$). It is unclear whether this equality holds without additional assumptions on $\agg$ and $\cA$ because, in general, it is not possible to take an indicator function out of an upper-semicontinuous hull. For example, consider the simple situation where $\Omega=\{\omega\}$ and $d=2$. In this case, we have the identification $(\cX,\cX')=(\R^2,\R^2)$. Consider the concave and positively homogeneous function $f$ and the convex closed set $\cC$ defined by
\[
f=-\ind_{\cD}, \ \ \ \ \cD=\{z\in\R^2 \,; \ 0\leq z_1<z_2\}\cup\{(0,0)\}, \ \ \ \ \cC=\{z\in\R^2 \,; \ z_1=z_2=1\}=\{(1,1)\}.
\]
Then, it is easy to see that
\[
\ucl(f-\ind_{\cC}) = -\ind_{\emptyset} \neq -\ind_{\{(1,1)\}} = \ucl(f)-\ind_{\cC}.
\]
\end{remark}


\subsection{Conditions for the identity $\sigma_{\agg^{-1}(\cA)}=\alpha$ to hold}

We know from Theorem~\ref{theo: representation support counterimage} that the support function of the systemic acceptance set $\agg^{-1}(\cA)$ always coincides with the upper semicontinuous hull of the penalty function $\alpha$. However, as illustrated by Example~\ref{ex:alpha_not_usc}, there are simple situations where the map $\alpha$ fails to be upper semicontinuous and, hence, the equality $\sigma_{\agg^{-1}(\cA)}=\alpha$ does not hold. In this subsection we establish a variety of sufficient conditions for this equality to hold. Clearly, one could also ask when $\sigma_{\agg^{-1}(\cA)}=\alpha^+$, which would automatically imply the statement for $\alpha$. While it is easy to find examples where this holds, none of the conditions in this section apply to $\alpha^+$.

\smallskip

As a first step, we highlight that the desired equality can be equivalently expressed in terms of a suitable minimax problem.

\begin{lemma}
\label{lem: minimax}
Let $Z\in\cX'$ and define a map $K:\cX\times\cE'\to[-\infty,\infty]$ by setting
\[
K_Z(X,W) := \sigma_\cA(W)+\E[\langle X,Z\rangle]-\E[\agg(X)W].
\]
The following statements are equivalent:
\begin{enumerate}[(a)]
  \item $\sigma_{\agg^{-1}(\cA)}=\alpha$.
  \item $\alpha$ is $\sigma(\cX',\cX)$-upper semicontinuous.
  \item For every $Z\in\cX'$ we have
\[
\inf_{X\in\cX}\sup_{W\in\cE'}K_Z(X,W) = \sup_{W\in\cE'}\inf_{X\in\cX}K_Z(X,W).
\]
\end{enumerate}
\end{lemma}
\begin{proof}
The equivalence between {\em (a)} and {\em (b)} is clear by Theorem~\ref{theo: representation support counterimage}. To establish equivalence with {\em (c)}, fix $Z\in\cX'$ and note that
\[
\alpha(Z)=\sup_{W\in\cE'}\inf_{X\in\cX}K_Z(X,W)
\]
by definition of $\alpha$. It remains to show that
\[
\sigma_{\agg^{-1}(\cA)}(Z)=\inf_{X\in\cX}\sup_{W\in\cE'}K_Z(X,W).
\]
Consider the auxiliary functions $f_Z:\cX\to(-\infty,\infty]$ defined by
\[
f_Z(X) := \E[\langle X,Z\rangle]+\ind_{\agg^{-1}(\cA)}(X)
\]
and $F_Z:\cX\times\cE\to(-\infty,\infty]$ defined by
\[
F_Z(X,U) := \E[\langle X,Z\rangle]+\ind_{\cA-\agg(X)}(U).
\]
Note that for every $X\in\cX$ the map $F_Z(X,\cdot)$ is convex and lower semicontinuous and satisfies
\begin{eqnarray*}
(F_Z(X,\cdot))^\ast(W)
&=&
\sup_{U\in\cE}\{\E[UW]-F_Z(X,U)\} \\
&=&
\sup_{U\in\cE,\,U+\agg(X)\in\cA}\{\E[UW]-\E[\langle X,Z\rangle]\} \\
&=&
\sup_{V\in\cE}\{\E[(V-\agg(X))W]-\E[\langle X,Z\rangle]\} \\
&=&
-\sigma_\cA(-W)-\E[\langle X,Z\rangle]+\E[\agg(X)(-W)] \\
&=&
-K_Z(X,-W)
\end{eqnarray*}
for every $W\in\cE'$. As $F_Z(X,0)=f_Z(X)$ for every $X\in\cX$, we can apply Fenchel-Moreau to get
\[
\sigma_{\agg^{-1}(\cA)}(Z)
=
\inf_{X\in\cX}f_Z(X)
=
\inf_{X\in\cX}\sup_{W\in\cE'}\{\E[0W]-(F_Z(X,\cdot))^\ast(W)\}
=
\inf_{X\in\cX}\sup_{W\in\cE'}K_Z(X,W).
\]
This concludes the proof.
\end{proof}

\smallskip

The preceding lemma shows that, for every $Z\in\cX'$, the identity $\sigma_{\agg^{-1}(\cA)}(Z)=\alpha(Z)$ is equivalent to the existence of a saddle value for the function $K_Z$. Unfortunately, the standard minimax theorems, see e.g.\ Fan~\cite{Fan1953}, rely on compactness assumptions that do not hold in our setting. The remainder of this subsection is devoted to showing a number of situations where the identity holds or, equivalently, the above minimax problem has a solution.


\subsubsection*{The linear case}

We start by proving the desired equality in the simple case where the impact map is given by the aggregated, or consolidated, capital position of all the $d$ financial institutions. In this case, there is no restriction on the acceptance set.

\begin{proposition}
Assume that $\cX_i=\cE$ for every $i\in\{1,\dots,d\}$. If $\agg(X)=\sum_{i=1}^dX_i$ for every $X\in\cX$, then $\alpha=\sigma_{\agg^{-1}(\cA)}$.
\end{proposition}
\begin{proof}
First of all, we show that for every $Z\in\cX'_+$ we have
\[
\sigma_{\agg^{-1}(\cA)}(Z) =
\begin{cases}
\sigma_\cA(Z_1) & \mbox{if \ $Z_1=\cdots=Z_d$},\\
-\infty & \mbox{otherwise}.
\end{cases}
\]
To see this, assume first that $\probp(Z_i>Z_j)>0$ for some distinct $i,j\in\{1,\dots,d\}$ and for every $n\in\N$ define a random vector $X^n\in\cX$ by
\[
X^n_k =
\begin{cases}
-n\mathbbm{1}_{\{Z_i>Z_j\}} & \mbox{if \ $k=i$},\\
n\mathbbm{1}_{\{Z_i>Z_j\}} & \mbox{if \ $k=j$},\\
0 & \mbox{otherwise}.
\end{cases}
\]
Since $\agg(X^n)=0\in\cA$ for every $n\in\N$, we clearly have
\[
\sigma_{\agg^{-1}(\cA)}(Z) \leq \inf_{n\in\N}\E[\langle X^n,Z\rangle] = \inf_{n\in\N}n\E[\mathbbm{1}_{\{Z_i>Z_j\}}(Z_j-Z_i)] = -\infty.
\]
Next, assume that $Z_1=\cdots=Z_d$ and note that, in this case, we have
\[
\sigma_{\agg^{-1}(\cA)}(Z) = \inf_{X\in\agg^{-1}(\cA)}\E[\agg(X)Z_1] = \sigma_\cA(Z_1).
\]
This proves the above claim. Now, for every $Z\in\cX'_+$ note that
\[
\alpha(Z)
=
\sup_{W\in\cB(\cA)}\left\{\sigma_\cA(W)+\inf_{X\in\cX}\E\left[\sum_{i=1}^dX_i(Z_i-W)\right]\right\}
=
\begin{cases}
\sigma_\cA(Z_1) & \mbox{if \ $Z_1=\cdots=Z_d\in\cB(\cA)$},\\
-\infty & \mbox{otherwise}.
\end{cases}
\]
This yields the desired assertion.
\end{proof}


\subsubsection*{The conic case}

Next, we deal with the case where $\agg$ is positively homogeneous and $\cA$ is a cone. In this case, we first show that $\alpha$ is given by a suitable indicator function and provide a general sufficient condition for the equality between $\sigma_{\agg^{-1}(\cA)}$ and $\alpha$. At a later stage, we apply this general condition to a variety of concrete situations.

\begin{lemma}
\label{lem: conic case}
Assume that $\agg$ is positively homogeneous and $\cA$ is a cone. Then, we have $\alpha=-\ind_\cD$ for
\[
\cD := \{Z\in\cX'_+ \,; \ \exists W\in\cB(\cA) \,:\, \E[\langle X,Z\rangle]\geq\E[\agg(X)W], \ \forall X\in\cX\}.
\]
\end{lemma}
\begin{proof}
Clearly, for every $Z\in\cD$ there exists $W_Z\in\cB(\cA)$ such that
\[
\inf_{X\in\cX}\{\E[\langle X,Z\rangle]-\E[\agg(X)W_Z]\} = \E[\langle 0,Z\rangle]-\E[\agg(0)W_Z] = 0.
\]
As a result, for every $Z\in\cD$ we have $0 \geq \alpha(Z) \geq \sigma_\cA(W_Z)+0 = 0$, showing that $\alpha(Z)=0$. Now, fix $Z\in\cX'\setminus\cD$ and observe that, for every $W\in\cB(\cA)$, we find $X_W\in\cX$ such that $\E[\langle X_W,Z\rangle]<\E[\agg(X_W)W]$. Then,
\begin{eqnarray*}
\inf_{X\in\cX}\{\E[\langle X,Z\rangle]-\E[\agg(X)W]\}
&\leq&
\inf_{n\in\N}\{\E[\langle nX_W,Z\rangle]-\E[\agg(nX_W)W]\} \\
&=&
\inf_{n\in\N}\{n(\E[\langle X_W,Z\rangle]-\E[\agg(X_W)W])\} \\
&=&
-\infty.
\end{eqnarray*}
This implies that $\alpha(Z)=-\infty$ and concludes the proof.
\end{proof}

\smallskip

\begin{lemma}
\label{lemma:conic_Lbd(X*)>0}
Assume that $\agg$ is positively homogeneous and $\cA$ is a cone. Moreover, assume that $\agg(e)\in\R_+\setminus\{0\}$ and that $\cB(\cA)\cap\{W\in L^1(\R) \,; \ \|W\|_1\leq1\}$ is $\sigma(\cE',\cE)$-compact. Then, $\sigma_{\agg^{-1}(\cA)}=\alpha$.
\end{lemma}
\begin{proof}
Recall that $\sigma_{\agg^{-1}(\cA)}=\alpha$ holds if and only if $\alpha$ is $\sigma(\cX',\cX)$-upper semicontinuous. Hence, by Lemma~\ref{lem: conic case}, it suffices to show that $\cD$ is $\sigma(\cX',\cX)$-closed. To this effect, take a net $(Z_\gamma)\subset\cD$ converging to some $Z\in\cX'$ in the topology $\sigma(\cX',\cX)$. Note that $Z\in\cX'_+$. By definition of $\cD$, for each $\gamma$ we find $W_\gamma\in\cB(\cA)$ such that
\[
\E[\langle X,Z_\gamma\rangle] \geq \E[\agg(X)W_\gamma]
\]
for every $X\in\cX$. To establish the desired closedness, it is enough to show that $(W_\gamma)$ admits a subnet that converges to some element of $\cB(\cA)$ in the topology $\sigma(\cE',\cE)$. Note that $\cB(\cA)=\{\sigma_\cA\geq0\}$ by conicity of $\cA$, showing that $\cB(\cA)$ is $\sigma(\cE',\cE)$-closed. Since $\cB(\cA)\subset\cE'_+$, we see that
\[
\E[\langle e,Z_\gamma\rangle] \geq \E[\agg(e)W_\gamma] \geq 0,
\]
or equivalently
\[
\frac{\E[\langle e,Z_\gamma\rangle]}{\agg(e)} \geq \E[W_\gamma] \geq 0,
\]
for every $\gamma$. Since $\E[\langle e,Z_\gamma\rangle]\to\E[\langle e,Z\rangle]$, the net $(W_\gamma)$ is bounded in $L^1(\R)$ and, hence, by using the compactness assumption, it admits a convergent subnet in the topology $\sigma(\cE',\cE)$. In view of the $\sigma(\cE',\cE)$-closedness of $\cB(\cA)$, we infer that the limit belongs to $\cB(\cA)$. This concludes the proof.
\end{proof}

\smallskip

The next proposition describes a number of situations where we can ensure the above compactness condition and, thus, we can establish that $\sigma_{\agg^{-1}(\cA)}=\alpha$.

\begin{proposition}
\label{prop:ph_alphausc}
Assume that $\agg$ is positively homogeneous and $\cA$ is a cone. Moreover, assume that $\agg(e)\in\R_+\setminus\{0\}$. Then, $\sigma_{\agg^{-1}(\cA)}=\alpha$ in each of the following cases:
\begin{enumerate}[(i)]
    \item $\Omega$ is finite.
    \item $\cA$ is polyhedral, i.e.\ there exist $W_1,\dots,W_n\in\cE'_+$ and $a_1,\dots,a_n\in\R$ such that
\[
\cA = \bigcap_{i=1}^n\{U\in\cE \,; \ \E[UW_i]\geq a_i\}.
\]
    \item $\cA$ is induced by Expected Shortfall, i.e.\ there exists $\lambda\in(0,1)$ such that
\[
\cA = \{U\in\cE \,; \ \ES_\lambda(U)\leq0\}.
\]
\end{enumerate}
\end{proposition}
\begin{proof}
{\em (i)} In the case that $\Omega$ is finite, the space $\cE'$ is finite dimensional and the compactness condition in Lemma~\ref{lemma:conic_Lbd(X*)>0} is clearly satisfied because $\cB(\cA)=\{\sigma_\cA\geq0\}$ is always $\sigma(\cE',\cE)$-closed.

\smallskip

{\em (ii)} If $\cA$ is polyhedral, then it is easy to see that $\cB(\cA)$ is a finitely-generated convex cone, i.e.\ there exist $W_1,\dots,W_n\in\cE'_+$ such that
\[
\cB(\cA) = \left\{\sum_{i=1}^n\lambda_iW_i \,; \ \lambda_1,\dots,\lambda_n\in[0,\infty)\right\}.
\]
Note that for all $\lambda_1,\dots,\lambda_n\in[0,\infty)$ we have
\[
\left\|\sum_{i=1}^n\lambda_iW_i\right\|_1 = \sum_{i=1}^n\lambda_i\|W_i\|_1.
\]
As a result, $\cB(\cA)\cap\{W\in L^1(\R) \,; \ \|W\|_1\leq1\}$ is easily seen to be $\sigma(\cE',\cE)$-compact and we can apply Lemma~\ref{lemma:conic_Lbd(X*)>0} to get the desired result.

\smallskip

{\em (iii)} If $\cA$ is induced by Expected Shortfall as in Example~\ref{ex:alpha_not_usc}, then
\[
\cB(\cA) = \left\{W\in\cE'_+ \,; \ W\leq\frac{1}{\lambda}\E[W]\right\}.
\]
As a result, we easily see that
\[
\cB(\cA)\cap\{W\in L^1(\R) \,; \ \|W\|_1\leq1\} \subset \{W\in L^\infty_+(\R) \,; \ W\leq\lambda^{-1}\}.
\]
Since the set $\cB(\cA)\cap\{W\in L^1(\R) \,; \ \|W\|_1\leq1\}$ is $\sigma(L^\infty(\R),L^1(\R))$-closed, it follows from the Banach-Alaoglu Theorem that it is even $\sigma(L^\infty(\R),L^1(\R))$-compact. As $\cE\subset L^1(\R)$, we automatically have $\sigma(\cE',\cE)$-compactness and we may conclude by applying Lemma~\ref{lemma:conic_Lbd(X*)>0}.
\end{proof}


\subsubsection*{The case where the image of $\agg$ intersects the interior of $\cA$}

As a final step, we follow Rockafellar~\cite{Rockafellar_Conjugate} to establish the identity $\sigma_{\agg^{-1}(\cA)}=\alpha$ under a suitable interiority condition, which also appears in Armenti et al.~\cite{ArmentiCrepeyDrapeauPapapantoleon2018} and Biagini et al.~\cite{BiaginiFouqueFrittelliMeyerBrandis2019b}.

\begin{proposition}
\label{prop:S(X)interiorA}
\begin{enumerate}[(i)]
  \item If there exists $X^\ast\in\cX$ such that $\agg(X^\ast)$ belongs to the $\sigma(\cE,\cE')$-interior of $\cA$, then $\alpha=\sigma_{\agg^{-1}(\cA)}$.
  \item Assume that $\cE'$ is the norm dual of $\cE$. If there exists $X^\ast\in\cX$ such that $\agg(X^\ast)$ belongs to the norm interior of $\cA$, then $\alpha=\sigma_{\agg^{-1}(\cA)}$.
\end{enumerate}
\end{proposition}
\begin{proof}
{\em (i)} By assumption, we find a $\sigma(\cE,\cE')$-neighborhood of zero $\cU\subset\cE$ such that $\agg(X^\ast)+\cU\subset\cA$. Now, fix an element $Z\in\cX'$ and define a map $\psi_Z:\cE\to[-\infty,\infty]$ by setting
\[
\psi_Z(U) := \inf_{X\in\cX}F_Z(X,U).
\]
Here, we have adopted the notation introduced in the proof of Lemma~\ref{lem: minimax}. It is easy to verify that $F_Z$ is jointly convex and, hence, $\psi_Z$ is convex. Note that
\[
\psi_Z(U) \leq F_Z(X^\ast,U) = \E[\langle X^\ast,Z\rangle]
\]
for every $U\in\cU$, so that $\psi_Z$ is bounded from above on $\cU$. In view of Lemma~\ref{lem: minimax}, the desired assertion follows from Theorem 17 in Rockafellar~\cite{Rockafellar_Conjugate} (by taking $\varphi=\psi_Z$ and $F=F_Z$ in the notation of that result).

\smallskip

{\em (ii)} Since the norm topology on $\cE$ is compatible with our bilinear form on $\cE\times\cE'$ under the assumption that $\cE'$ is the norm dual of $\cE$, we can repreat the same argument as in {\em (i)} by exploiting the fact that Theorem 17 in~\cite{Rockafellar_Conjugate} holds under any compatible topology.
\end{proof}


\section{``First aggregate, then allocate''-type systemic risk measures}

In this short section we turn to systemic risk measures of ``first aggregate, then allocate'' type and their dual representations. Throughout the section we fix an admissible impact map $\agg$ and an admissible acceptance set $\cA$.


\subsection{The systemic risk measure $\widetilde{\rho}$}

``First aggregate, then allocate''-type systemic risk measures are defined as follows.

\begin{definition}
We define a map $\widetilde{\rho}:\cX\to[-\infty,\infty]$ by setting
\[
\widetilde{\rho}(X) = \inf\{m\in\R \,; \ \agg(X)+m\in\cA\}.
\]
\end{definition}

\smallskip

The difference with respect to $\rho$ is that, instead of injecting capital into the system in order to reach an acceptable level of systemic risk, one looks at the minimum level of the chosen systemic risk indicator that ensures acceptability. In particular, if the impact function is expressed in monetary terms, then $\widetilde{\rho}(X)$ can be interpreted as a bail-out cost for the ``aggregated position'' $\agg(X)$. For a thorough presentation of this type of systemic risk measures we refer to the literature cited in the introduction.

\smallskip

In what follows, we exploit the fact that $\widetilde{\rho}$ can be expressed as the composition between the impact map and the standard cash-additive risk measure $\rho_\cA:\cE\to[-\infty,\infty]$ given by
\[
\rho_\cA(X):=\inf\{m\in\R \,; \ X+m\in\cA\}.
\]

\smallskip

The next result records the key properties of $\widetilde{\rho}$. In particular, differently from the systemic risk measure $\rho$, we show that $\widetilde{\rho}$ is always lower semicontinuous under our standing assumptions on the impact map and the acceptance set.

\begin{proposition}
\label{prop: properties rho tilde}
The systemic risk measure $\widetilde{\rho}$ is convex, $\sigma(\cX,\cX')$-lower semicontinuous, and satisfies $\widetilde{\rho}(0)\leq0$. Moreover, $\widetilde{\rho}$ is proper if and only if $\widetilde{\rho}(0)>-\infty$ if and only if $\cA\cap(-\R_+)\neq-\R_+$.
\end{proposition}
\begin{proof}
Convexity is clear by composition. To show lower semicontinuity, note that $\rho_\cA$ is $\sigma(\cE,\cE')$-lower semicontinuous by the $\sigma(\cE,\cE')$-closedness of $\cA$. Now, take $r\in\R$ and note that
\[
\{X\in\cX \,; \ \widetilde{\rho}(X)\leq r\} = \agg^{-1}(\{U\in\cE \,; \ \rho_\cA(U)\leq r\}).
\]
Following the argument in the proof of Proposition~\ref{prop: properties rho} we can show that the above set is $\sigma(\cX,\cX')$-closed, which delivers the desired lower semicontinuity. To show properness, observe first that $\widetilde{\rho}(0)\leq0$ because $\agg(0)=0\in\cA$. The above equivalence can now be established as in the proof of Proposition~\ref{prop: properness rho}.
\end{proof}


\subsection{The dual representation of $\widetilde{\rho}$}

The purpose of this subsection is to derive a dual representation of $\widetilde{\rho}$ and to compare it with the dual representation of $\rho$. In this case, the acceptability test is performed on $\agg(X)$ and the acceptance set is $\cA$. This suggests to rely on the dual representation of $\rho_\cA$ in order to achieve in a straightforward way the desired dual representation of $\widetilde{\rho}$. The following maps are the fundamental ingredients of the desired representation.

\begin{definition}
We define two maps $\widetilde{\alpha},\widetilde{\alpha}^+:\cX'\to[-\infty,+\infty]$ by setting
\[
\widetilde{\alpha}(Z) := \sup_{W\in\cB(\cA),\,\E[W]=1}\Big\{\sigma_\cA(W)+\inf_{X\in\cX}\{\E[\langle X,Z\rangle]-\E[\agg(X)W]\}\Big\},
\]
\[
\widetilde{\alpha}^+(Z) := \sup_{W\in\cB(\cA)\cap(\cE'_{++}\cup\{0\}),\,\E[W]=1}\Big\{\sigma_\cA(W)+\inf_{X\in\cX}\{\E[\langle X,Z\rangle]-\E[\agg(X)W]\}\Big\}.
\]
\end{definition}

\smallskip

\begin{remark}
\label{rem: alpha tilde K}
The above maps belong to the class of maps $\widetilde{\alpha}_\cK:\cX'\to[-\infty,+\infty]$ defined by
\[
\widetilde{\alpha}_\cK(Z) := \sup_{W\in\cK,\,\E[W]=1}\Big\{\sigma_\cA(W)+\inf_{X\in\cX}\{\E[\langle X,Z\rangle]-\E[\agg(X)W]\}\Big\},
\]
where $\cK$ is a convex cone in $\cB(\cA)$ such that $\lambda\cK+(1-\lambda)\cB(\cA)\subset\cK$ for every $\lambda\in[0,1]$; see also Remark~\ref{rem: alpha K}. This will allow us to prove properties for $\widetilde{\alpha}$ and $\widetilde{\alpha}^+$ simultaneously. In fact, all properties of $\widetilde{\alpha}$ and $\widetilde{\alpha}^+$ we will consider are shared by the entire class.
\end{remark}

\smallskip

Before we establish the desired dual representation we highlight some relevant properties of the above maps and point out their relationship with the penalty functions $\alpha$ and $\alpha^+$. Here, we denote by $\dom(\widetilde{\alpha})$ the domain of finiteness of $\widetilde{\alpha}$ (similarly for $\widetilde{\alpha}^+$). Moreover, we denote by $\cl$ the closure operator with respect to the topology $\sigma(\cX',\cX)$.

\begin{proposition}
\label{prop: properties alpha tilde}
The maps $\widetilde{\alpha},\widetilde{\alpha}^+:\cX'\to[-\infty,\infty]$ satisfy the following properties (the statements about $\widetilde{\alpha}^+$ require that $\cB(\cA)\cap\cE'_{++}\neq\emptyset$):
\begin{enumerate}[(i)]
  \item $\widetilde{\alpha}$ and $\widetilde{\alpha}^+$ take values in the interval $[-\infty,0]$.
  \item $\widetilde{\alpha}$ and $\widetilde{\alpha}^+$ are concave.
  \item $\dom(\widetilde{\alpha}^+)\subset\dom(\widetilde{\alpha})\subset \cl(\dom(\widetilde{\alpha}^+))\subset\cX'_+$.
  \item $\alpha$ is the smallest positively homogeneous map dominating $\widetilde{\alpha}$, i.e.\ for every $Z\in\cX'$
\[
\alpha(Z) = \sup_{\lambda>0}\frac{\widetilde{\alpha}(\lambda Z)}{\lambda}.
\]
  \item $\alpha^+$ is the smallest positively homogeneous map dominating $\widetilde{\alpha}^+$, i.e.\ for every $Z\in\cX'$
\[
\alpha^+(Z) = \sup_{\lambda>0}\frac{\widetilde{\alpha}^+(\lambda Z)}{\lambda}.
\]
\end{enumerate}
\end{proposition}
\begin{proof}
{\em (i)-(ii), (iv)-(v)} Let $\cK\subset\cB(\cA)$ be a convex cone as in Remark~\ref{rem: alpha tilde K}. It is clear that
\[
\alpha_\cK(Z) = \sup_{\lambda>0}\frac{\widetilde{\alpha}_\cK(\lambda Z)}{\lambda}
\]
for every $Z\in\cX'$. In particular, $\widetilde{\alpha}_\cK\leq\alpha_\cK$. It follows from the proof of Proposition~\ref{prop:properties_alpha} that $\widetilde{\alpha}_\cK$ takes value into $[-\infty,0]$. Moreover, the proof of the concavity of $\alpha_\cK$ in that result can be repeated to establish the concavity of $\widetilde{\alpha}_\cK$. The desired assertions follow by taking $\cK=\cB(\cA)$ and $\cK=\cB(\cA)\cap(\cE'_{++}\cup\{0\})$.

\smallskip

{\em (iii)} The assertion can be proved by repeating the proof of the corresponding statement in Proposition~\ref{prop:properties_alpha}.
\end{proof}

\smallskip

We record the announced dual representation of $\widetilde{\rho}$ in the next result.

\begin{theorem}
\label{theo: dual representation rho tilde}
(i) If $\widetilde{\rho}$ is proper, then we have
\[
\widetilde{\rho}(X) = \sup_{Z\in\cX'_+}\{\widetilde{\alpha}(Z)-\E[\langle X,Z\rangle]\}
\]
for every $X\in\cX$. The supremum can be restricted to $\cX'_{++}$ provided that $\dom(\widetilde{\alpha})\cap\cX'_{++}\neq\emptyset$.

\smallskip

(ii) Assume that $\cB(\cA)\cap\cE'_{++}\neq\emptyset$. If $\widetilde{\rho}$ is proper, then we have
\[
\widetilde{\rho}(X) = \sup_{Z\in\cX'_+}\{\widetilde{\alpha}^+(Z)-\E[\langle X,Z\rangle]\}
\]
for every $X\in\cX$. The supremum can be restricted to $\cX'_{++}$ provided that $\dom(\widetilde{\alpha}^+)\cap\cX'_{++}\neq\emptyset$.
\end{theorem}
\begin{proof}
Let $\cK\subset\cB(\cA)$ be a convex cone as in Remark~\ref{rem: alpha tilde K}. Note that the dual representation~\eqref{eq: external characterization} applied to $\cA$ yields
\begin{equation}
\label{eq: representation rho tilde}
\cA = \bigcap_{W\in\cK}\{U\in\cE \,; \ \E[UW]\geq\sigma_\cA(W)\} = \bigcap_{W\in\cK,\,\E[W]=1}\{U\in\cE \,; \ \E[UW]\geq\sigma_\cA(W)\},
\end{equation}
where we used the positive homogeneity of $\sigma_\cA$ (together with the fact that $\cK\subset\cE'_+$). As a result, for every $U\in\cE$ we get
\[
\rho_\cA(U) = \sup_{W\in\cK,\,\E[W]=1}\{\sigma_\cA(W)-\E[UW]\}.
\]
Using the notation introduced in the proof of Theorem~\ref{prop:Lambda-1(A)_dualrepr}, we immediately get
\begin{eqnarray*}
\widetilde{\rho}(X)
&=&
\sup_{W\in\cK,\,\E[W]=1}\{\sigma_\cA(W)-\E[\agg(X)W]\} \\
&=&
\sup_{W\in\cK,\,\E[W]=1}\sup_{Z\in\cX'_+}\{\sigma_\cA(W)-\E[\langle X,Z\rangle]+(\varphi_W)^\bullet(Z)\} \\
&=&
\sup_{Z\in\cX'_+}\sup_{W\in\cK,\,\E[W]=1}\{\sigma_\cA(W)-\E[\langle X,Z\rangle]+(\varphi_W)^\bullet(Z)\} \\
&=&
\sup_{Z\in\cX'_+}\{\widetilde{\alpha}_\cK(Z)-\E[\langle X,Z\rangle]\}
\end{eqnarray*}
for every $X\in\cX$. If, in addition, $\dom(\widetilde{\alpha}_\cK)\cap\cX'_{++}\neq\emptyset$, then we get
\[
\widetilde{\rho}(X) = \sup_{Z\in\cX'_{++}}\{\widetilde{\alpha}_\cK(Z)-\E[\langle X,Z\rangle]\}
\]
for every $X\in\cX$ by the same argument used to reduce the domain of the supremum in the proof of Theorem~\ref{theo: dual representation}. The desired assertions now follow by taking $\cK=\cB(\cA)$ and $\cK=\cB(\cA)\cap(\cE'_{++}\cup\{0\})$.
\end{proof}

\smallskip

\begin{remark}
(i) As in Remark~\ref{rem: fenchel moreau for rho}, we highlight the link between the dual representation in Theorem~\ref{theo: dual representation rho tilde} and the standard Fenchel-Moreau representation. We claim that, if $\widetilde{\rho}$ is proper, then
\[
\widetilde{\rho}^\ast(Z) = -\ucl(\widetilde{\alpha})(-Z) = -\ucl(\widetilde{\alpha}^+)(-Z)
\]
for every $Z\in\cX'$ (where the last equality holds provided that $\cB(\cA)\cap\cE'_{++}\neq\emptyset$). Here, we have denoted by $\ucl(\widetilde{\alpha})$ the $\sigma(\cX',\cX)$-upper semicontinuous hull of $\widetilde{\alpha}$ (similarly for $\widetilde{\alpha}^+$). To see this, note first that
\[
\widetilde{\rho}(X) = \sup_{Z\in\cX'}\{\widetilde{\alpha}(Z)-\E[\langle X,Z\rangle]\} = \sup_{Z\in\cX'}\{\ucl(\widetilde{\alpha})(Z)-\E[\langle X,Z\rangle]\}
\]
for every $X\in\cX$. The left-hand side equality holds because $\widetilde{\alpha}=-\infty$ outside $\cX'_+$ by Proposition~\ref{prop: properties alpha tilde}. The right-hand side equality follows from Theorem 2.3.1 in Z\u{a}linescu~\cite{Zalinescu2002}. Since $\ucl(\widetilde{\alpha})$ is concave and $\sigma(\cX',\cX)$-upper semicontinuous, the desired claim is a consequence of the Fenchel-Moreau Theorem. The argument for $\widetilde{\alpha}^+$ is identical.

\smallskip

(ii) The dual elements in the above representation can be identified with $d$-dimensional vectors of probability measures on $(\Omega,\cF)$ that are absolutely continuous (or equivalent) with respect to $\probp$ up to a normalizing vector that collects their expectations. This allows to express the above representation in terms of probability measures. Indeed, for every $w\in\R^d_+$ define
\[
\cQ^w(\probp) := \{\probq\in\cQ(\probp) \,; \ \probq_i=\probp \ \mbox{if} \ w_i=0, \ \forall i\in\{1,\dots,d\}\}, \ \ \ \ \cQ^w_e(\probp)=\cQ_e(\probp)\cap\cQ^w(\probp),
\]
where we have used the notation from Remark~\ref{rem: fenchel moreau for rho}. Then, if $\widetilde{\rho}$ is proper, we easily see that
\[
\widetilde{\rho}(X) = \sup_{w\in\R^d_+,\,\probq\in\cQ^w(\probp),\,\frac{d\probq}{d\probp}\in\cX'}\bigg\{
\widetilde{\alpha}\bigg(w_1\frac{d\probq_1}{d\probp},\dots,w_d\frac{d\probq_d}{d\probp}\bigg)-
\sum_{i=1}^dw_i\E_{\probq_i}[X_i]\bigg\}
\]
for every $X\in\cX$. We can replace $\cQ^w(\probp)$ by $\cQ^w_e(\probp)$ in the above supremum provided that $\dom(\widetilde{\alpha})\cap\cX'_{++}\neq\emptyset$. The same holds with $\widetilde{\alpha}^+$ instead of $\widetilde{\alpha}$ (provided that $\cB(\cA)\cap\cE'_{++}\neq\emptyset$).
\end{remark}

\smallskip

The condition $\dom(\widetilde{\alpha})\cap\cX'_{++}\neq\emptyset$ is needed to restrict the domain in the above dual representation to strictly-positive dual elements (similarly for $\widetilde{\alpha}^+$). We conclude this section by providing a sufficient condition for this to hold; see also Proposition~\ref{prop: cond str pos}.

\begin{proposition}
Assume that $\cX_i=\cE$ for every $i\in\{1,\dots,d\}$. Moreover, suppose that $\cB(\cA)\cap\cE'_{++}\neq\emptyset$ and there exist $a\in(0,\infty)$ and $b\in\R$ such that
\[
\agg(X) \leq a\sum_{i=1}^dX_i+b
\]
for every $X\in\cX$. Then, $\dom(\widetilde{\alpha}^+)\cap\cX'_{++}\neq\emptyset$ (and, a fortiori, $\dom(\widetilde{\alpha})\cap\cX'_{++}\neq\emptyset$).
\end{proposition}
\begin{proof}
Take $W\in\cB(\cA)\cap\cE'_{++}$ and note that we can always assume that $\E[W]=1$ by conicity of $\cB(\cA)$. Setting $Z=(aW,\dots,aW)\in\cX'_{++}$, we easily see that
\[
\widetilde{\alpha}(Z)
\geq
\widetilde{\alpha}^+(Z)
\geq
\sigma_\cA(W)+\inf_{X\in\cX}\{\E[\langle X,Z\rangle]-\E[\agg(X)W]\}
\geq
\sigma_\cA(W)-b\E[W]
>
-\infty.
\]
This delivers the desired assertion.
\end{proof}


\section{Risk measures based on univariate utility functions}
\label{sect: utility}

In this final section we provide a simple proof of the dual representation of shortfall risk measures, see Theorem 4.115 in F\"{o}llmer and Schied~\cite{FoellmerSchied2016}, that uses our general strategy to obtain dual representations. For ease of comparison, we focus on bounded positions.

\smallskip

Throughout the entire section we fix a nonconstant, concave, increasing function $u:\R\to\R$, which is interpreted as a standard von Neumann-Morgenstern utility function. We fix $u_0\in\R$ such that $u(x)>u_0$ for some $x\in\R$ and define a map $\rho_u:L^\infty\to[-\infty,\infty]$ by
\[
\rho_u(X) := \inf\{m\in\R \,; \ \E[u(X+m)]\geq u_0\}.
\]

\smallskip

\begin{theorem}
The risk measure $\rho_u$ is convex and $\sigma(L^\infty,L^1)$-lower semicontinuous. Moreover,
\[
\rho_u(X) = \sup_{\probq\ll\probp}\left\{\E_\probq[-X]+\sup_{\lambda>0}\left\{ \frac{1}{\lambda}\left(u_0+\E\left[u^\bullet\left(\lambda\frac{d\probq}{d\probp}\right)\right]\right) \right\}\right\}
\]
for every $X\in L^\infty$.
\end{theorem}
\begin{proof}
It is well-known that $\rho_u$ is convex and $\sigma(L^\infty,L^1)$-lower semicontinuous. To establish the above representation, note that $\rho_u$ can be viewed as a ``first allocate, then aggregate''-type systemic risk measure corresponding to the specifications
\[
d=1, \ \ \ (\cX,\cX')=(L^\infty,L^1), \ \ \ (\cE,\cE')=(L^\infty,L^1), \ \ \ \agg(X)=u(X), \ \ \ \cA=\{U\in L^\infty \,; \ \E[U]\geq u_0\}.
\]
First of all, note that $\cB(\cA)=\R_+$ and $\sigma_\cA(\lambda)=\lambda u_0$ for every $\lambda\in\R_+$. Since $\cB(\cA)\cap\cE'_{++}$ is nonempty, we can work with $\alpha^+$; see Definition~\ref{def: penalties}. It follows from Proposition~\ref{prop: alpha in the Lambda case} that
\begin{eqnarray*}
\alpha^+(Z)
&=&
\sup_{W\in\cB(\cA)\cap\cE'_{++}}\bigg\{\sigma_\cA(W)+
\E\bigg[u^\bullet\bigg(\frac{Z}{W}\bigg)W\bigg]\bigg\} \\
&=&
\sup_{\lambda>0}\bigg\{\lambda u_0+ \E\bigg[u^\bullet\bigg(\frac{Z}{\lambda}\bigg)\lambda\bigg]\bigg\} \\
&=&
\sup_{\lambda>0}\bigg\{\frac{1}{\lambda}(u_0+\E[u^\bullet(\lambda Z)])\bigg\}
\end{eqnarray*}
for every nonzero $Z\in\cX'_+$. The representation of $\agg^{-1}(\cA)$ in Theorem~\ref{prop:Lambda-1(A)_dualrepr} yields
\[
\agg^{-1}(\cA) = \bigcap_{Z\in\cX'_+\setminus\{0\}}\{X\in\cX \,; \ \E[XZ]\geq\alpha^+(Z)\} = \bigcap_{\probq\ll\probp}\bigg\{X\in\cX \,; \ \E_\probq[X]\geq\alpha^+\bigg(\frac{d\probq}{d\probp}\bigg)\bigg\},
\]
where we used that $\dom(\alpha^+)\subset\cX'_+$ and that $\alpha^+$ is positively homogeneous; see Proposition~\ref{prop:properties_alpha}. It remains to observe that
\[
\rho_u(X) = \inf\{m\in\R \,; \ X+m\in\agg^{-1}(\cA)\}
\]
for every $X\in\cX$.
\end{proof}


\end{document}